\documentclass[11pt,reqno]{article}
\usepackage{amsmath,amsfonts,amsthm}
\usepackage[margin=1in]{geometry}
\usepackage[normalem]{ulem}

\usepackage[hyphens]{url}
\usepackage{hyperref}
\usepackage{graphicx}
\usepackage{verbatim}
\usepackage{enumitem} % to control indention of "itemize" by [leftmargin=0.1in]
\usepackage{color}
\usepackage{tikz}
\numberwithin{equation}{section}

\newtheorem{thm}{Theorem}[section]

\newtheorem{coro}{Corollary}[section]
\newtheorem{lem}{Lemma}[section]
\newtheorem{prop}{Proposition}[section]

\newtheorem{defn}{Definition}[section]
\newtheorem{rem}{Remark}[section]
\newtheorem{eg}{Example}[section]

\renewcommand{\P}{\mathbb{P}}
\newcommand{\Q}{\mathbb{Q}}
\newcommand{\R}{\mathbb{R}}
\newcommand{\E}{\mathbb{E}}

\newcommand{\N}{\mathbb{N}}

\newcommand{\cB}{\mathcal{B}}
\newcommand{\cH}{\mathcal{H}}

\newcommand{\cM}{\mathcal{M}}

\newcommand{\cF}{\mathcal{F}}

\newcommand{\cG}{\mathcal{G}}

\newcommand{\M}{\mathcal{M}}

\newcommand{\A}{\mathcal{A}}

\newcommand{\eps}{\varepsilon}

\newcommand{\fP}{\mathfrak{P}}

\newcommand{\bA}{\text{\bf A}}

\newcommand{\USC}{\operatorname{USC}}
\newcommand{\USA}{\operatorname{USA}}

\newcommand{\nada}[1]{}
\definecolor{gb}{rgb}{0, 0.2, 0.8}

\title{Generalized Duality for Model-Free Superhedging given Marginals%\\--- a Portfolio Constraint Approach}
}
\author{Arash Fahim\thanks{Florida State University, Department of Mathematics, Tallahassee, FL 32306-4510, USA, email: \texttt{fahim@math.fsu.edu}. Partially supported by Florida State University CRC FYAP (315-81000-2424) and National Science Foundation (DMS-1209519).}
\and
Yu-Jui Huang\thanks{
University of Colorado, Department of Applied Mathematics, Boulder, CO 80309-0526, USA, email: \texttt{yujui.huang@colorado.edu}. Partially supported by National Science Foundation (DMS-1715439) and the University of Colorado (11003573).}
 \and Saeed Khalili\thanks{
University of Colorado, Department of Mathematics, Boulder, CO 80309-0395, USA, email: \texttt{saeed.khalili@colorado.edu}.}
}
\date{\today}
\begin{document}
\maketitle

\begin{abstract}
In a discrete-time financial market, a generalized duality is established for model-free superhedging, given marginal distributions of the underlying asset. Contrary to prior studies, we do not require contingent claims to be upper semicontinuous, allowing for upper semi-analytic ones. The generalized duality stipulates an extended version of risk-neutral pricing. To compute the model-free superhedging price, one needs to find the supremum of expected values of a contingent claim, evaluated {\it not} directly under martingale (risk-neutral) measures, but along sequences of measures that converge, in an appropriate sense, to martingale ones. To derive the main result, we first establish a portfolio-constrained duality for upper semi-analytic contingent claims, relying on Choquet's capacitability theorem. As we gradually fade out the portfolio constraint, the generalized duality emerges through delicate probabilistic estimations.
\end{abstract}

\textbf{MSC (2010):} 
%49K21, % Optimality conditions:	Problems involving relations other than differential equations
60G42,  %	Martingales with discrete parameter
%60J05,  %	Discrete-time Markov processes on general state spaces
%60J27,  %Continuous-time Markov processes on discrete state spaces
%91A13, % Games with infinitely many players
91G20,  %	Derivative securities
91G80.  %	Financial applications of other theories (stochastic control, calculus of variations, PDE, SPDE, dynamical systems)
%93E20. % Optimal stochastic control
\smallskip


%%%%%%%%%%%%%%%%%%%%%%%%%%%%%%%%%%%%%%%%%%%%%%%%%%%%
%%%%%%%%%%%%%%%%%%%%%%%%%%%%%%%%%%%%%%%%%%%%%%%%%%%%

\section{Introduction}
%Take $\R_+:=\{x\in\R:x>0\}$. 
Given a finite time horizon $T\in \N$ with $T\ge 2$, let $\Omega:=\R_+^T=[0,\infty)^T$ be the path space and $S$ be the canonical process, i.e. $S_{t}(x_1,x_2,...,x_T)=x_{t}$ for all $(x_1,x_2,...,x_T)\in\Omega$. We denote by $\fP(\Omega)$ the set of all probability measures on $\Omega$. %We denote by $\F=\{\cF_t\}_{t=0}^T$ the natural filtration generated by $S$. 
For all $t=1,...,T$, let $\mu_t$ be a probability measure on $\R_+$ that has %by a probability measure $\mu_t$ on $\R_+$ with 
finite first moment; namely,
\begin{equation}\label{finite moment}
m(\mu_t):=\int_{\R_+}yd\mu_t(y)<\infty. %\quad \text{\rm for all}\ t=1,...,T.
\end{equation}
The set of admissible probability measures on $\Omega$ is given by
\begin{equation}\label{defn:Pi}
\Pi:=\left\{\Q\in\fP(\Omega): \Q\circ(S_t)^{-1} = \mu_{t},\ \forall t=1,...,T\right\},
\end{equation}
%where $\text{Law}^\Q(S_t)$ denotes the law of $S_t$ under $\Q\in\fP(\Omega)$. 
which is known to be nonempty, convex, and compact under the topology of weak convergence, thanks to \cite[Proposition 1.2]{Kellerer84}. We further consider
\begin{equation}\label{defn:cM}
\cM:=\left\{\Q\in\Pi: S\ \hbox{is a $\Q$-martingale}\right\}.
\end{equation}
Note that $\cM\neq\emptyset$ if and only if $\mu_1,...,\mu_T$ possess the same finite first moment and increase in the convex order (i.e. $\int_{\R_+}f d\mu_1\le \int_{\R_+}f d\mu_2\le...\le \int_{\R_+}f d\mu_T$, for convex $f:\R_+\to\R$); see \cite{Strassen65}. We will assume $\cM\neq\emptyset$ throughout this paper. %, and note that $\cM$ is also compact under the topology of weak convergence, as shown in \cite[Proposition 2.4]{BHP13}. 

The current setup is motivated by a financial market that involves a risky asset, represented by $S$, and abundant tradable options written on it. For instance, if the tradable options at time 0 include vanilla call options, with payoff $(S_{t}-K)^+$, for all $t=1,\cdots,T$ and $K\ge 0$, then the current market prices $C(t,K)$ of these call options already prescribe the distribution of $S_{t}$, for each $t=1,...,T$, under any pricing (martingale) measure.\footnote{By \cite[Proposition 2.1]{HR12}, for each fixed $t$, as long as $K\mapsto C(t,K)$ is convex and nonnegative, $\lim_{K\downarrow 0+}\partial_KC(t,K)\ge-1$, and $\lim_{K\to\infty}C(t,K)=0$, the relation ``$\E_\Q[(S_{t}-K)^+]=C(t,K)$ for all $K\ge 0$'' determines the distribution of $S_{t}$. That is, $\Pi$ in \eqref{defn:Pi} can be expressed as $\left\{\Q\in\fP(\Omega):\E_\Q[(S_{t}-K)^+]=C(t,K),\ \forall t=1,\cdots,T\ \hbox{and}\ K\ge 0\right\}$.
}

A path-dependent contingent claim $\Phi:\Omega\to\R$ can be superhedged by trading the underlying $S$ and holding options available at time 0. Specifically, let $\cH$ be the set of $\Delta=\{\Delta_t\}_{t=1}^{T-1}$ with $\Delta_{t}:\R^t_+\to \R$ Borel measurable for all $t=1,..., T-1$. Each $\Delta\in\cH$ represents a self-financing (dynamic) trading strategy. The resulting change of wealth over time along a path $x=(x_1,...,x_T)\in \Omega$ is given by
\[
(\Delta\cdot x)_t := \sum_{i=1}^{t-1}\Delta_i(x_1,...,x_i)\cdot(x_{i+1}-x_{i}),\quad \ \hbox{for}\ t=2,..., T.
\]
In addition, by writing $\mu=(\mu_1,...,\mu_T)$, we denote by $L^1(\mu)$ the set of $u=(u_1,...,u_T)$ where $u_{t}:\R_+\to\R$  is $\mu_t$-integrable for all $t=1,...,T$. Each $u\in L^1(\mu)$ represents a collection of options with different maturities. 
A semi-static superhedge of $\Phi$ consists of some $\Delta\in\cH$ and $u\in L^1(\mu)$ such that
\begin{equation}\label{superhedge}
\Psi_{u,\Delta}(x):=\sum_{t=1}^T  u_{t}(x_{t}) + (\Delta\cdot x)_T \ge \Phi(x),\quad  \hbox{for all}\ x=(x_1,...,x_T)\in\Omega.
\end{equation}
Such superhedging is {\it model-free}: the terminal wealth $\Psi_{u,\Delta}$ is required to dominate $\Phi$ on {\it every} path $x\in\Omega$, instead of $\P$-a.e. $x\in\Omega$ for some probability $\P$. This is distinct from the standard model-based approach: classically, one first specifies a model, or physical measure, $\P$ for the financial market, and then superhedges a contingent claim $\P$-a.s. With the {\it pointwise} relation \eqref{superhedge}, no matter which $\P$ materializes, $\Psi_{u,\Delta}\ge \Phi$ must hold $\P$-a.s.
There is then no need to specify a physical measure $\P$ a priori, which prevents any model misspecification. 

The corresponding model-free superhedging price of $\Phi$ is defined by
\begin{equation}\label{D}
D(\Phi):=\inf\left\{\mu(u): u\in L^1(\mu)\ \hbox{satisfies}\ \exists \Delta\in\cH\ \hbox{s.t.}\ \Psi_{u,\Delta}(x)\ge \Phi(x)\ \forall x\in\Omega\right\},
\end{equation}
where $\mu(u):=\sum_{t=1}^T\int_{\R_+}u_{t}d\mu_{t}$. To characterize $D(\Phi)$, the minimal cost to achieve \eqref{superhedge}, Beiglb\"{o}ck, Henry-Labord\'{e}re, and Penkner \cite{BHP13} introduce the martingale optimal transport problem
\begin{equation}\label{primal}
P(\Phi):=\sup_{\Q\in\cM}\E_\Q[\Phi]. 
\end{equation}
When $\Phi$ is upper semicontinuous, denoted by $\Phi\in\operatorname{USC}(\Omega)$, and grows linearly, $D(\Phi)$ coincides with $P(\Phi)$.

\begin{prop}[Corollary 1.1, \cite{BHP13}]\label{prop:duality usc}
Given $\Phi\in\operatorname{USC}(\Omega)$ for which there exists $K>0$ such that
\begin{equation}\label{Phi bounded below}
\Phi(x)\le K(1+x_1+\cdots+x_T)\quad \forall x=(x_1,\cdots,x_T)\in\Omega, 
\end{equation} 
we have $D(\Phi) = P(\Phi)$. % and $D^N(\Phi)=P^N(\Phi)$ for all $N\in\N$. % and  $P^{N,\ell}(\Phi)=D^{N,\ell}(\Phi)$ for all $N\in\N$ and $\ell>0$. 
\end{prop}

Model-free superhedging given marginals, pioneered by Hobson \cite{Hobson98}, has traditionally focused on specific forms of contingent claims; see e.g. \cite{BHR01}, \cite{HLW05-QF}, \cite{LW05}, \cite{CDDV08}, and \cite{CO11}. The main contribution of \cite{BHP13} is to allow for general, albeit upper semicontinuous, contingent claims, via the superhedging duality in Proposition~\ref{prop:duality usc}. In deriving this duality, \cite{BHP13} uses upper semicontinuity only once for a minimax argument. It is tempting to believe that upper semicontinuity is only a technical condition that can eventually be relaxed. 

This is, however, {\it not} the case. While the model-free duality given marginals in \cite{BHP13} has been widely studied and enriched by now (see \cite{DS14}, \cite{ABPS16}, \cite{FH16}, and \cite{CHO16}, among others), the requirement of upper semicontinuity stands still. Recently, Beiglb\"{o}ck, Nutz, and Touzi \cite{BNT17} has shown that, in fact, upper semicontinuity {\it cannot} be relaxed. They provide a counterexample where $\Phi$ is lower, but not upper, semicontinuous and the duality $D(\Phi)=P(\Phi)$ fails. 
To restore the duality, \cite{BNT17} modifies the definition of $D(\Phi)$ in \eqref{D} in a quasi-sure way: the inequality $\Psi_{u,\Delta}\ge \Phi$ is required to hold {\it not} pointwise, but $\M$-quasi surely; that is, $\Psi_{u,\Delta}\ge \Phi$ holds outside of a set that is $\P$-null for all $\P\in\M$. This quasi-sure modification successfully yields the duality $D_{\operatorname{qs}}(\Phi) = P(\Phi)$ for Borel measurable $\Phi$, where $D_{\operatorname{qs}}(\Phi)$ denotes the modified $D(\Phi)$ as described above. This is done in \cite{BNT17} for the two-period model (i.e. $T=2$), and in Nutz, Stebegg, and Tan \cite{NST19} for the multi-period case (i.e. $T\in \N$). 

In this paper, we approach the failure of $D(\Phi)=P(\Phi)$ from an opposite angle. We keep the definition of $D(\Phi)$ as in \eqref{D}, and investigate how $P(\Phi)$ should be modified to get a general duality for Borel measurable $\Phi$ and beyond. This has two motivations in terms of both theory and applications. 

From the theoretical point of view, the pointwise relation \eqref{superhedge} is inherited from the optimal transport theory: the dual problem in the Monge-Kantorovich duality is almost identical to $D(\Phi)$, except that it involves the simpler pointwise relation $\sum_{t=1}^T  u_{t}(x_{t}) \ge \Phi(x)$ (i.e. without the term $(\Delta\cdot x)_T$ in \eqref{superhedge}); see \cite{Kellerer84}. That is, $D(\Phi)$ naturally extends the classical dual problem from optimal transport to the more general setting we focus on. Finding the primal problem corresponding to this extended dual is of great theoretical interest in itself. 

More crucially, as $D(\Phi)$ represents precisely the minimal cost for model-free superhedging, if we modify its definition, although a duality can be obtained (as in \cite{BNT17} and \cite{NST19}), it will no longer adhere to the model-free superhedging context, thereby losing its financial relevance. In fact, there are two different applications here. In the context of optimal transport, $\Phi$ is a payoff function that assigns a reward to each transportation path $x=(x_1,...,x_T)\in \Omega$, and every $\Q\in\cM$ is an admissible %(randomized) 
transportation plan. The goal is to maximize reward from transportation, i.e. to attain $P(\Phi)$ in \eqref{primal}---the perspective taken by \cite{BNT17} and \cite{NST19}. Our goal, by contrast,  is to minimize the cost of model-free superhedging; %For the financial application of model-free superhedging, the focus of this paper, 
all developments should then be centered around $D(\Phi)$ in \eqref{D}. 

Instead of dealing with $D(\Phi)$ directly, we impose, somewhat artificially, portfolio constraints. For any $N\in\N$, we consider
\begin{equation}\label{H^N}
\mathcal{H}^N :=\{\Delta\in\mathcal{H}:\  |\Delta_{t}| \le N,\ \forall t=1,\cdots,T-1\},
\end{equation} 
and define $D^N(\Phi)$ as in \eqref{D}, with $\cH$ therein replaced by $\cH^N$. That is, $D^N(\Phi)$ is a portfolio-constrained model-free superhedging price. Thanks to the general duality in Fahim and Huang \cite{FH16}, the corresponding primal problem $P^N(\Phi)$ can be identified, and there is no duality gap (i.e. $D^N(\Phi)=P^N(\Phi)$) when $\Phi$ is upper semicontinuous. The first major contribution of this paper, Theorem~\ref{thm:P^N=D^N}, shows that this portfolio-constrained duality actually holds generally for upper semi-analytic $\Phi$. Specifically, by treating $D^N$ and $P^N$ as functionals, we derive appropriate upward and downarrow continuity (Sections~\ref{subsec:P^N} and \ref{subsec:D^N}). Choquet's capacitability theorem can then be invoked to extend $D^N(\Phi)=P^N(\Phi)$ from upper semicontinuous $\Phi$ to upper semi-analytic ones. %; see  for details.

Note that the portfolio bound $N\in\N$ is indispensable here. In the technical result Lemma~\ref{lem:compactness for D}, the compactness of the space of semi-static strategies $(u,\Delta)\in L^1(\mu)\times\cH^N$ is extracted from the bound $N\in\N$, under an appropriate weak topology. Such compactness then gives rise to the upward continuity of $D^N$; see Proposition~\ref{prop:D^N upward}. As opposed to this, $D$ in \eqref{D}, when viewed as a functional, does not possess the desired upward continuity. This prevents a direct application of Choquet's capacitability theorem to the unconstrained duality $D(\Phi)=P(\Phi)$ in Proposition~\ref{prop:duality usc}; see Remark~\ref{rem:cannot apply} for details.

By taking $N\to\infty$ in the constrained duality $D^N(\Phi)=P^N(\Phi)$, we obtain a new characterization of $D(\Phi)$, for upper semi-analytic $\Phi$; see Theorem~\ref{thm:main}, the main result of this paper. This new characterization asserts a generalized version of risk-neutral pricing. To find the model-free superhedging price $D(\Phi)$, we need to compute expected values of $\Phi$, but {\it not} directly under risk-neutral (martingale) measures $\Q\in\cM$. As prescribed by Theorem~\ref{thm:main}, we should consider sequences of measures $\{\Q_n\}_{n\in\N}$ that converge to $\cM$ appropriately, and compute the limiting expected values, i.e. $\limsup_{n\to\infty}\E_{\Q_n}[\Phi]$. The supremum of these limiting expected values then characterizes $D(\Phi)$. For the special case where $\Phi$ is upper semicontinuous, these limiting expected values can be attained by measures $\Q\in\cM$, as shown in Proposition~\ref{prop:tP=P}. The generalized duality in Theorem~\ref{thm:main} thus reduces to one that involves solely measures in $\cM$, recovering the classical duality in Proposition~\ref{prop:duality usc}. 

In deriving the generalized duality in Theorem~\ref{thm:main} from the constrained one $D^N(\Phi)=P^N(\Phi)$, one needs the relation $\lim_{N\to\infty} D^N(\Phi) = D(\Phi)$. This is equivalent to $ D^\infty(\Phi) = D(\Phi)$, where $D^\infty(\Phi)$ is defined as in \eqref{D}, with $\cH$ therein replaced by
\begin{equation}\label{H^infty}
\cH^\infty := \{\Delta\in\mathcal{H}:\  \Delta_{t}\ \hbox{is bound},\ \forall t=1,\cdots,T-1\}.
\end{equation}
This turns out to be highly nontrivial, and is established through delicate probabilistic estimations; see Proposition~\ref{thm: sec.main} for details. Such a relation is economically intriguing in itself: it states that restricting to bounded trading strategies does not increase the cost of model-free superhedging. To the best of our knowledge, this harmless restriction to bounded strategies has not been identified in the literature under such generality. 

The rest of the paper is organized as follows. Section~\ref{sec:setup} introduces the main result of this paper, a generalized duality that characterizes $D(\Phi)$, for upper semi-analytic $\Phi$. %and outlines how it will be proved. 
Section~\ref{sec:P^N=D^N} establishes a portfolio-constrained duality for upper semi-analytic contingent claims, by using Choquet's capacity theory. Section~\ref{pf.thm.main} derives an unconstrained duality for upper semi-analytic contingent claims, as the limiting case of the constrained one in Section~\ref{sec:P^N=D^N}; this completes the proof of the main result.

\subsection{Notation}\label{subsec:notation} 
Let $Y=\R^t_+$ for some $t=1,2,...,T$. % be the extended real line. %For a Hausdorff topological space $Y$, 
We denote by $\cG(Y)$ the set of all functions from $\Omega$ to $\R$. Moreover, let $\USA(Y)$, $\cB(Y)$, and $\USC(Y)$ be the sets of functions in $\cG(Y)$ that are upper semi-analytic, Borel measurable, and upper semicontinuous, respectively. 
% functions from $\Omega$ to $\R^*$, by $\USC(\Omega)$ the set of upper semicontinuous functions from $\Omega$ to $\R^*$, and by $\USA(\Omega)$ the set of upper semi-analytic functions from $\Omega$ to $\R^*$.
Throughout this paper, for any $\Phi\in\cG(\Omega)$ and $\Q\in\Pi$, we will interpret $\E_\Q[\Phi]$ as the {\it outer expectation} of $\Phi$. %, defined as $\inf\{\E_{\Q}[\psi]: \psi\ge\Phi,\text{ $\psi\in\cB(\Omega)$}\}$. 
When $\Phi$ is actually Borel measurable, it reduces to the standard expectation of $\Phi$. 

For any $u\in L^1(\mu)$, we will write $\oplus u(x) := \sum_{t=1}^T u_{t}(x_{t})$ for $x=(x_1,...,x_T)\in\Omega$ and $\mu(u) :=\sum_{t=1}^T\int_{\R_+}u_{t}d\mu_{t}$, as specified below \eqref{D}.

%Given a subset $\A(\R^T_+)$ of $\cG(\R^T_+)$, we denote by
%\begin{itemize}
%\item $\A^L(\R^T_+)$ the set of functions $\Phi\in\A(\R^T_+)$ for which there exists $K>0$ such that
%\begin{equation}\label{Phi bounded above}
%\Phi(x)\le K(1+x_1+\cdots+x_T),\quad \forall x=(x_1,\cdots,x_T)\in\R_+^T.
%\end{equation}
%\item $\A_L(\R^T_+)$ the set of functions $\Phi\in\A(\R^T_+)$ for which there exists $K>0$ such that
%\begin{equation}\label{Phi bounded below}
%\Phi(x)\ge -K(1+x_1+\cdots+x_T),\quad \forall x=(x_1,\cdots,x_T)\in\R_+^T.
%\end{equation} 
%\end{itemize}

%%%%%%%%%%%%%%%%%%%%%%%
%%%%%%%%%%%%%%%%%%%%%%%

\section{The Main Result}\label{sec:setup}
\subsection{Preliminaries}
%and notice that $\cF_T$ coincides with the Borel $\sigma$-algebra of $\Omega$. 
%Note that we take $S$ to be one-dimensional simply for the ease of presentation; all results in this paper can be generalized to the multi-dimensional case in a straightforward manner; see \cite{FH16} for details.
Given $N\in\N$, recall $\cH^N$ defined in \eqref{H^N}. For each $\Q\in\Pi$, we introduce
\begin{align}\label{def: A_t}
\text{\bf A}^{N}_T(\Q):=\sup_{\Delta\in\cH^N}\E_\Q[(\Delta\cdot S)_T]=\sup_{\Delta\in\cH^N_c}\E_\Q[(\Delta\cdot S)_T],
\end{align}
where 
\begin{equation*}%\label{defn cS^infty}
\mathcal{H}^N_c :=\{\Delta\in\mathcal{H}^N:\  \Delta_{t}\ \hbox{is continuous},\ \forall t=1,\cdots,T-1\}.
\end{equation*}
Note that the reduction to continuous trading strategies in \eqref{def: A_t} is justified by \cite[Lemma 3.3]{FH16}. The set $\cM$ in \eqref{defn:cM} can be fully characterized by $\text{\bf A}^{N}_T(\Q)$ as follows. 

\begin{lem}\label{prop:cM equivalence}
$\Q\in\cM$ $\iff$ $\text{\bf A}^{1}_T(\Q)=0$ $\iff$ $\bA^N_T(\Q)=0$ for all $N\in\N$.
\end{lem}
%$\iff$ $\text{\bf A}^{\infty}_T(\Q)=0$ & $\iff$ $\bA^{N,\ell}_T(\Q)=0$  for all $N\in\N$ and $\ell>0$. 
\begin{proof}
By definition, $\bA^N_T(\Q) = N \bA^1_T(\Q)$. Thus, $\bA^1_T(\Q) =0$ if and only if $\bA^N_T(\Q) =0$ for all $N\in\N$. Now, by \eqref{def: A_t}, ``$\bA^N_T(\Q)=0$ for all $N\in\N$'' is equivalent to ``$\E_\Q[(\Delta\cdot S)_T]=0$ for all $\Delta\in\cH^N_c$, for any $N\in\N$''. The latter condition holds if and only if $\Q\in\cM$, by \cite[Lemma 2.3]{BHP13}. 
\end{proof}

Lemma~\ref{prop:cM equivalence} indicates that a pseudometric on $\Pi$ can be defined by
\begin{equation}\label{d}
d(\Q_1,\Q_2) := \left|\bA^1_T(\Q_1)-\bA^1_T(\Q_2)\right|,\quad \forall \Q_1,\Q_2\in\Pi. 
\end{equation}
It is only a pseudometric, but not a metric, because $d(\Q_1,\Q_2)=0$ does not necessarily imply $\Q_1=\Q_2$. We can turn it into a metric by considering equivalent classes induced by $d$. Specifically, we say $\Q_1,\Q_2\in\Pi$ are equivalent (denoted by $\Q_1\sim\Q_2$) if $d(\Q_1,\Q_2)=0$, or $\bA^1_T(\Q_1) = \bA^1_T(\Q_2)$. Equivalent classes are then defined by $[\Q] := \{\Q'\in\Pi : d(\Q',\Q)=0\}$ for all $\Q\in \Pi$. 
On the quotient space $\Pi^*:=\Pi/\sim\ = \{[\Q]:\Q\in\Pi\}$,
\begin{equation}\label{rho}
\rho([\Q_1],[\Q_2]) := d(\Q_1,\Q_2)
\end{equation}
defines a metric. 

\begin{rem}\label{rem:cM=[Q]}
In view of Lemma~\ref{prop:cM equivalence}, $\cM=[\Q]$ for any $\Q\in\cM$. 
\end{rem}

\begin{rem}
Instead of the pseudometric on $\Pi$ in \eqref{d}, one can consider the semi-norm 
\begin{align*}%\label{def: rho}
\|Q\| := \sup_{\Delta \in \mathcal H^1 } \int_{\Omega} (\Delta\cdot S)_T\ dQ
\end{align*}
defined on the vector space $\mathcal K := \{ Q : Q \ \text{is a signed measure on $\Omega$}\}$. When we restrict the semi-norm to $\Pi\subset \mathcal K$, we have $\|\Q\|=0$ if and only $\Q\in\cM$ (thanks to Lemma~\ref{prop:cM equivalence}). This can be used to define a metric equivalent to \eqref{rho}.
\end{rem}

To state the main result of this paper, Theorem~\ref{thm:main} below, we need to consider a sequence $\{\Q_N\}_{N\in\N}$ in $\Pi$ that converge to $\cM$ under the metric $\rho$; that is, by Remark~\ref{rem:cM=[Q]}, 
\[
\rho([\Q_N],\cM) = \rho([\Q_N],[\Q]) \to 0\quad \hbox{as}\ N\to\infty,\quad \forall \Q\in \cM. 
\]
For simplicity, this will be denoted by $\Q_N\overset{\rho}{\to}\cM$. As $\Q_N\overset{\rho}{\to}\cM$ is equivalent to $\bA^1_T(\Q_N)\to 0$, by \eqref{rho} and \eqref{d}, they will be used interchangeably throughout the paper.  

Crucially, $\Q_N\overset{\rho}{\to}\cM$ entails weak convergence to $\cM$ (up to a subsequence). 

\begin{lem}\label{lem:weak convergence to cM}
Consider $\{\Q_N\}_{N\in\N}$ in $\Pi$ such that $\Q_N\overset{\rho}{\to}\cM$. For any subsequence $\{\Q_{N_k}\}_{k\in\N}$ that converges weakly, it must converge weakly to some $\Q_*\in\cM$.   
\end{lem}

\begin{proof}
%Since $\Pi$ is weakly compact, $\Q_N$ must converge to (up to a subsequence) $\Q_*\in\Pi$. 
Let $\Q_*\in\Pi$ denote the probability measure to which $\Q_{N_k}$ converges weakly. First, recall that  $\Q_N\overset{\rho}{\to}\cM$ is equivalent to $\bA^1_T(\Q_N)\to 0$, which in turn implies $\bA^1_T(\Q_{N_k})\to 0$. Next, for any $\Delta\in\cH^1_c$, since $|(\Delta\cdot x)_T|\le h(x):=x_1+2(x_2+...+x_{T-1})+x_T$, we deduce from  \cite[Lemma 4.3]{Villani-book-09} that $\Q\mapsto \E_\Q[(\Delta\cdot S)_T]$ is continuous under the topology of weak convergence. It follows that 
\[
\Q\mapsto \bA^{1}_T(\Q)= \sup_{\Delta\in\cH^1_c}\E_\Q[(\Delta\cdot S)_T]
\]
is lower semicontinuous under the topology of weak convergence. Hence,
\[
\bA^1_T(\Q_*) \le \liminf_{k\to\infty} \bA^1_T(\Q_{N_k}) = 0.
\]   
We then conclude $\bA^1_T(\Q_*)=0$, which implies $\Q_*\in\cM$ thanks to Lemma~\ref{prop:cM equivalence}. 
\end{proof}

%%%%%%%%%%%%%%%%%%%%%%%

\subsection{The Generalized Duality}
Now, we are ready to present the main result of this paper. 

\begin{thm}\label{thm:main}
For any $\Phi\in\USA(\Omega)$ for which there exists $K>0$ such that 
\begin{equation}\label{Phi bounded}
|\Phi(x)|\le K(1+x_1+\cdots+x_T)\quad \forall x=(x_1,\cdots,x_T)\in\Omega,
\end{equation}
%such that \eqref{Phi bounded} is satisfied for some $K>0$,
we have
\begin{equation}\label{general duality}
D(\Phi) = \widetilde P(\Phi):=\sup_{}\left\{\limsup_{N\to\infty} \E_{\Q_N}[\Phi] :\Q_N\overset{\rho}{\to}\cM \right\}.
\end{equation} 
\end{thm}
When $\Phi$ is additionally upper semicontinuous, Theorem~\ref{thm:main} recovers the classical duality in Proposition~\ref{prop:duality usc}, as the next result demonstrates.

\begin{prop}\label{prop:tP=P}
For any $\Phi\in\operatorname{USC}(\Omega)$ that satisfies \eqref{Phi bounded}, $\widetilde P(\Phi)$ reduces to $P(\Phi)$ in \eqref{primal}. 
\end{prop}

\begin{proof}
For any $\Q\in\cM$, by taking $\Q_N := \Q$ for all $N\in\N$, the definition of $\widetilde P(\Phi)$ in \eqref{general duality} directly implies $\widetilde P(\Phi)\ge \E_\Q[\Phi]$. Taking supremum over $\Q\in\cM$ yields $\widetilde P(\Phi)\ge P(\Phi)$. 

On the other hand, take an arbitrary $\{\Q_N\}_{N\in\N}$ in $\Pi$ such that $\Q_N\overset{\rho}{\to}\cM$. For any $\eps>0$, there exists a subsequence $\{\Q_{N_k}\}_{k\in\N}$ such that 
\begin{equation}\label{11'}
\lim_{k\to\infty} \E_{\Q_{N_k}}[\Phi] \ge \limsup_{N\to\infty} \E_{\Q_{N}}[\Phi]-\eps.
\end{equation}
As $\Pi$ is compact (recall the explanation below \eqref{defn:Pi}), there is a further subsequence, which will still be denoted by $\{\Q_{N_k}\}_{k\in\N}$, that converges weakly to some $\Q_*\in\Pi$. By Lemma~\ref{lem:weak convergence to cM}, $\Q_*$ must belong to $\cM$. Now, as $\Phi$ is upper semicontinuous and satisfies \eqref{Phi bounded}, we deduce from \cite[Lemma 4.3]{Villani-book-09} and $\{\Q_{N_k}\}$ converging weakly to $\Q_*\in\cM$ that 
\[
\lim_{k\to\infty} \E_{\Q_{N_k}}[\Phi] \le \E_{\Q_*}[\Phi]\le P(\Phi). 
\]
This, together with \eqref{11'} and the arbitrariness of $\eps>0$, shows that $\limsup_{N\to\infty} \E_{\Q_{N}}[\Phi]\le P(\Phi)$. As $\{\Q_N\}_{N\in\N}$ such that $\Q_N\overset{\rho}{\to}\cM$ is arbitrarily chosen, we conclude that $\widetilde P(\Phi)\le P(\Phi)$.
\end{proof}

Theorem~\ref{thm:main} extends the standard wisdom for risk-neutral pricing. To find the model-free superhedging price $D(\Phi)$, one needs to compute expected values of $\Phi$, but not directly under risk-neutral (martingale) measures $\Q\in\cM$. Instead, one should consider, more generally, sequences of measures $\{\Q_N\}_{N\in\N}$ in $\Pi$ that converge appropriately to $\cM$, and compute the limiting expected values of $\Phi$. Only when $\Phi$ is continuous enough (i.e. upper semicontinuous) can we restrict our attention to solely martingale measures in $\cM$, as Proposition~\ref{prop:tP=P} indicates. 

The next example demonstrates explicitly that despite $D(\Phi)>P(\Phi)$, the generalized duality $D(\Phi) = \widetilde P(\Phi)$ holds.  

\begin{eg}\label{eg:eg 8.1 from BNT}
Let $T=2$ and $\mu_1 = \mu_2$ be the Lebesgue measure on $[0,1]$. Then $\cM$ contains one single measure $\P_0$, under which $(S_1,S_2)$ is uniformly distributed on $\{(x,y)\in[0,1]^2: x=y\}$. For the lower semicontinuous $\Phi(x_1,x_2):=1_{\{x_1\neq x_2\}}$, it is shown in \cite[Example 8.1]{BNT17} that $0=P(\Phi)<D(\Phi)=1$;
in addition, $(u^*_1,u^*_2,\Delta^*_1)\equiv (1,0,0)$ is an optimizer of $D(\Phi)$.
%they showed that if $h_1(x)$, $h_2(x)$, and $\Delta(x)$ such that 
%\[
%h_1(x_1)+h_2(x_2)+\Delta(x_1)(x_2-x_1)\ge \Phi(x_1,x_2),
%\]
%then 
%\[
%\int h_1(x)d\mu_1(x)+\int h_2(x)d\mu_2(x)\ge1,
%\]
%and the minimal superhedging price is obtained by chooseing $h_1\equiv1$, $h_2\equiv0$, and $\Delta(x)\equiv0$. As  a result, under any portfolio constraint, the superhedging price $D(\Phi)$ of $\Phi$ is equal to $1$.

We will show that $\widetilde P(\Phi) =1$. %Fix $N\in\N$. As $\Delta^*_1\in\cH^N$, $(u^*_1,u^*_2,\Delta^*_1)$ is also an optimizer of $D^N(\Phi)$, and thus $D^N(\Phi) = D(\Phi)=1$. On the other hand, since $P^N(\Phi)\le \sup_{\Q\in\Pi}\E_\Q[\Phi]\le 1$ by definition, it remains to show that $P^N(\Phi)\ge 1$. 
Consider a collection of probability measures $\{\Q_M\}_{M\in\N}$ on $[0,1]^2$, with the density function of each $\Q_M$ given by
\begin{equation}\label{density}
g(x_1,x_2)=M\sum_{i=0}^{M-1}1_{[\frac{i}{M},\frac{i+1}{M})^2}(x_1,x_2);
\end{equation}
see Figure~\ref{fig:payoff}. It can be checked by definition that $\Q_{M}\in\Pi$.  
\begin{figure}[ht!]
\centering
 \begin{tikzpicture}
\draw [-, thick] (0,0) -- (0,4) -- (4,4) -- (4,0) -- (0,0);
\draw [-, blue,thick, fill] (0,0) -- (0,.5) --(.5,.5) -- (.5,0) --(0,0);
\draw [-, blue,thick, fill] (.5,.5) -- (.5,1) --(1,1) -- (1,.5) --(.5,.5);
\draw [-, blue,thick, fill] (1,1) -- (1,1.5) --(1.5,1.5) -- (1.5,1) --(1,1);
\draw [-, blue,thick, fill] (1.5,1.5) -- (1.5,2) --(2,2) -- (2,1.5) --(1.5,1.5);
\draw [-, blue,thick, fill] (2,2) -- (2,2.5) --(2.5,2.5) -- (2.5,2) --(2,2);
\draw [-, blue,thick, fill] (2.5,2.5) -- (2.5,3) --(3,3) -- (3,2.5) --(2.5,2.5);
\draw [dotted, blue,thick, fill] (3,3) --(3.5,3.5);
\draw [-, blue,thick, fill] (3.5,3.5) -- (3.5,4) --(4,4) -- (4,3.5) --(3.5,3.5);
\draw [|-|,thick] (0,-.2) -- (.5,-.2) node at (.25,-.5) {$\frac{1}{M}$};
\end{tikzpicture}
\caption{Support of $\Q_{M}$.}
\label{fig:payoff}
\end{figure}
Observe that 
\begin{align*}
\bA^1_2(\Q_M) =  \sup_{\Delta_1 \in \mathcal H^1 }\E_{\Q_M}[\Delta_1\cdot(S_{2}-S_1)] &=  \sup_{\Delta_1 \in \mathcal H^1 } \E_{\Q_M}\left[\Delta_1 \cdot \left(\E_{\Q_M}[S_{2}\mid\cF_1]-S_1\right)\right]\\
&= \E_{\Q_M}[|\E_{\Q_M}[S_2\mid \cF_1]-S_1|],
\end{align*}
where $\cF_1$ denotes the $\sigma$-algebra generated by $S_1$, and the second line holds as the supremum is attained by taking $\Delta_1=\mbox{sgn}(\E_{\Q_M}[S_{2}\mid\cF_1]-S_1)$. As
\[
\E_{\Q_M}[S_2\mid S_1=x]=\sum_{i=0}^{M-1}\frac{2i+1}{2M}1_{\{[\frac{i}{M},\frac{i+1}{M})\}}(x). 
\]
we obtain
\begin{equation}\label{A^N_2}
\bA^{1}_2(\Q_M) = \E_{\Q_M}[|\E_{\Q_M}[S_2\mid S_1]-S_1|] = \sum_{i=0}^{M-1}\frac{1}{4M^2} = \frac{1}{4M}\to 0,\quad \hbox{as}\ M\to\infty. 
\end{equation}
That is, $\Q_M\overset{\rho}{\to}\cM$. It follows that $\widetilde P(\Phi) \ge \limsup_{M\to\infty} \E_{\Q_M}[\Phi] =1$, where the equality stems from $\E_{\Q_M}[\Phi]=1$ for all $M\in\N$, by the definition of $\Phi$ and \eqref{density}. As $\Phi\le 1$ readily implies $\widetilde P(\Phi)\le 1$, we conclude $\widetilde P(\Phi)=1=D(\Phi)$. 
\end{eg}

This paper is devoted to the derivation of Theorem~\ref{thm:main}. It will be done through a delicate two-step plan, to be carried out in detail in Sections \ref{sec:P^N=D^N} and \ref{pf.thm.main}. We give a brief outline as follows. 

For any $N\in\N$, recall $D^N(\Phi)$, the portfolio-constrained model-free superhedging price defined below \eqref{H^N}. Also, consider
\begin{equation}\label{primal^N}
%P^\infty(\Phi)&:=\sup_{\Q\in\Pi}\{\E_\Q[\Phi]- \bA^{\infty}_T(\Q)\}\\
P^N(\Phi):=\sup_{\Q\in\Pi}\{\E_\Q[\Phi]-\text{\bf A}^{N}_T(\Q)\}.
\end{equation}
As a direct consequence of Fahim and Huang \cite[Theorem 3.14]{FH16}, $D^N(\Phi)$ can be characterized, in the same spirit of Proposition~\ref{prop:duality usc}, as follows. 

\begin{prop}\label{prop:duality usc N}
Given $\Phi\in\operatorname{USC}(\Omega)$ that satisfies \eqref{Phi bounded below}, $D^N(\Phi)=P^N(\Phi)$ for all $N\in\N$. 
\end{prop}

Section \ref{sec:P^N=D^N} focuses on extending this portfolio-constrained duality to one that allows for {\it upper semi-analytic} $\Phi$. Intriguingly, by using Choquet's capacity theory, we will show that the same duality $D^N(\Phi)=P^N(\Phi)$ simply holds for upper semi-analytic $\Phi$; there is no need to adjust $P^N(\Phi)$. By taking $N\to\infty$, Section \ref{pf.thm.main} elaborates how $D^N(\Phi)=P^N(\Phi)$ turns into the desired duality \eqref{general duality}.

%%%%%%%%%%%%%%%%%%%%%%%%%%%%%
%%%%%%%%%%%%%%%%%%%%%%%%%%%%%

\section{Complete Duality under Portfolio Constraints}\label{sec:P^N=D^N}
Given $N\in\N$, the goal of this section is to establish the complete duality $D^N(\Phi)=P^N(\Phi)$ for upper semi-analytic $\Phi$. %, or more generally for all upper semi-analytic $\Phi$. 
As such a duality is known to hold for upper semicontinuous $\Phi$ (Proposition~\ref{prop:duality usc N}), our strategy is to treat $P^N$ and $D^N$ as functionals, and exploit their continuity properties.
 
Let us first recall the notion of a Choquet capacity. Recall also the notation in Section~\ref{subsec:notation}.

\begin{defn}\label{defn:choquet-capcity}
A functional $C:\cG(\Omega)\to\R$ is called a Choquet capacity associated with $\USC(\Omega)$ (or simply capacity) if it satisfies 
\begin{itemize}
\item[(i)] $C(\phi)\le C(\psi)$ if $\phi\le\psi$;
\item[(ii)] if $\phi_i\uparrow\phi$, then $\sup_{i\in\N}C(\Phi_i)=C(\Phi)$;
\item[(iii)] for any sequence $\{\phi_i\}$ in $\USC(\Omega)$ such that $\phi_i\downarrow\phi$, $\inf_{i\in\N}C(\Phi_i)=C(\Phi)$.
\end{itemize}
\end{defn}

Choquet's capacitability theorem (see \cite[Proposition 2.11]{Kellerer84} or \cite[Section 3]{Choquet59}) asserts a desirable continuity property of a capacity.

\begin{lem}\label{lem:capacity}
Let $C:\cG(\Omega)\to\R$ be a Choquet capacity associated with $\USC(\Omega)$. Then, for any $\Phi\in\USA(\Omega)$, 
\[
C(\Phi)=\sup\{C(\phi)\;:\;\phi\le\Phi~\text{ with $\phi\in\USC(\Omega)$}\}.
\]
Hence, if two capacities $C_1$ and $C_2$ coincide on $\USC(\Omega)$, they coincide on $\USA(\Omega)$.
\end{lem}

\begin{rem}
The original Choquet's capacitability theorem gives a more general result: if $C_1$ and $C_2$ are two Choquet capacities associated with a set of functions $\A$ and they coincide on functions in $\A$, then they coincide on $\A$-Suslin functions. Here, we take $\A=\USC(\Omega)$ in Definition~\ref{defn:choquet-capcity} and Lemma~\ref{lem:capacity}, and note that ``$\USC(\Omega)$-Suslin functions'' are simply ``upper semi-analytic functions ($\USA(\Omega)$)''; see  \cite[Proposition 2.13]{Kellerer84} and \cite[Definition 7.21]{Bertsekas-Shreve-book}.    
\end{rem}

%%%%%%%%%%%%%%%%%%%%%%%%%%%%%%%%%%%%%%%%%%%%%%%%%%%%%%

\subsection{Continuity of $P^N$}\label{subsec:P^N}

\begin{prop}\label{prop:P^N upward}
Consider $\{\Phi_i\}_{i\in\N}$ in $\cG(\Omega)$ for which there exists $K>0$ such that for each $i\in\N$,
\begin{equation}\label{Phi bounded below'}
\Phi_i(x)\ge -K(1+x_1+\cdots+x_T),\quad \forall x=(x_1,\cdots,x_T)\in\Omega.
\end{equation}  
If $\Phi_i\uparrow\Phi$, then
\begin{align*}
%\sup_{i\in\N}P(\Phi_i)&=P(\Phi),\quad 
\sup_{i\in\N}P^N(\Phi_i)=P^N(\Phi),\quad \forall N\in\N.
\end{align*}
\end{prop}

\begin{proof}
Since $\Phi_1$ satisfies \eqref{Phi bounded below'}, the monotone convergence theorem for outer expectation gives $\E_{\Q}[\Phi_i]\uparrow\E_{\Q}[\Phi]$, for all $\Q\in\Pi$. By changing the order of two supremums, we get 
\[
\sup_{i\in\N}P^N(\Phi_i) = \sup_{\Q\in\Pi}\sup_{i\in\N} \left(\E_\Q[\Phi_i]-\text{\bf A}^{N}_T(\Q)\right)= \sup_{\Q\in\Pi}\left(\E_\Q[\Phi]-\text{\bf A}^{N}_T(\Q)\right)=P^N(\Phi),
\]
for each $N\in\N$.
\end{proof}

\begin{prop}\label{prop:P^N downward}
Consider $\{\Phi_i\}_{i\in\N}$ in $\USC(\Omega)$ for which there exists $K>0$ such that \eqref{Phi bounded below} is satisfied for each $\Phi_i$. If $\Phi_i\downarrow\Phi$, then
\[
%\inf_{i\in\N}P(\Phi_i)=P(\Phi),\quad  
\inf_{i\in\N}P^N(\Phi_i)=P^N(\Phi),\quad \forall N\in\N.
\]
\end{prop}

\begin{proof}
Fix $N\in\N$. As $\Phi_i\downarrow\Phi$ clearly implies $\inf_{i\in\N}P^N(\Phi_i)\ge P^N(\Phi)$, we focus on proving the ``$\le$'' relation.  Assume $\inf_{i\in\N}P^N(\Phi_i)>-\infty$, otherwise the proof would be trivial. For any $\delta<\inf_{i\in\N} P^N(\Phi_i)$, define 
\[
\cM^N(\Phi_i,\delta):=\{\Q\in\Pi: \E_\Q[\Phi_i]-\bA^{N}_T(\Q)\ge\delta\}\quad\hbox{for all}\ N\in\N.
\]
We intend to show that $\cM^N(\Phi_i,\delta)$ is compact under the topology of weak convergence. As $\Pi$ is compact (recall the explanation below \eqref{defn:Pi}), it suffices to prove that $\cM^N(\Phi_i,\delta)$ is closed. Since $\Phi_i$ is upper semicontinuous and satisfies \eqref{Phi bounded below}, we deduce from \cite[Lemma 4.3]{Villani-book-09} that $\Q\mapsto \E_\Q[\Phi_i]$ is upper semicontinuous under the topology of weak convergence. On the other hand, %for each $\Delta\in\cH^N_c$, 
by the same argument in the proof of Lemma~\ref{lem:weak convergence to cM}, 
%since $|(\Delta\cdot x)_T|\le h(x):=N[x_1+2(x_2+...+x_{T-1})+x_T]$, we deduce from  \cite[Lemma 4.3]{Villani-book-09} that $\Q\mapsto \E_\Q[(\Delta\cdot S)_T]$ is continuous. It follows that 
%\[
%\Q\mapsto \bA^{N}_T(\Q)= \sup_{\Delta\in\cH^N_c}\E_\Q[(\Delta\cdot S)_T]
%\]
$\Q\mapsto \bA^{N}_T(\Q)$ is lower semicontinuous under the topology of weak convergence. As a result, $\Q\mapsto \E_\Q[\Phi_i]-\bA^{N}_T(\Q)$ is upper semicontinuous, which gives the desired closedness of $\cM^N(\Phi_i,\delta)$. 

Now, since $\{\cM^N(\Phi_i,\delta)\}_{i\in\N}$ is a nonincreasing sequence of compact sets, $\bigcap_{i=1}^\infty\cM^N(\Phi_i,\delta)\neq\emptyset$. Take $\tilde\Q\in\bigcap_{i=1}^\infty\cM^N(\Phi_i,\delta)$, and observe that 
\[
P^N(\Phi)\ge \E^{\tilde\Q}[\Phi]-\bA^{N}_T(\tilde\Q)=\lim_{i\to\infty}\E^{\tilde\Q}[\Phi_i]-\bA^{N}_T(\tilde\Q)\ge\delta,
\]
where the equality follows from the reverse monotone convergence theorem, applicable here as \eqref{Phi bounded below} is satisfied for each $\Phi_i$, and the last inequality results from the definition of $\cM^N(\Phi_i,\delta)$. With $\delta<\inf_{i\in\N} P^N(\Phi_i)$ arbitrarily chosen, we conclude $\inf_{i\in\N}P^N(\Phi_i)\le P^N(\Phi)$. 
\end{proof}

%The following corollary follows from Propositions \ref{prop:P^N upward} and \ref{prop:P^N downward}.
%\begin{coro}
%The functionals $P$ and $P^N$ are Choquet capacities.
%\end{coro}

%%%%%%%%%%%%%%%%%%%%%%%%%%%%%%%%%%%%%%%%%%%%%%%%%%%%%%%%%%%%%%%

\subsection{Continuity of ${D}^N$}\label{subsec:D^N}
The downward continuity of $D^N$ is a consequence of Propositions~\ref{prop:duality usc N} and \ref{prop:P^N downward}.

\begin{prop}\label{prop:D^N downward}
Consider $\{\Phi_i\}_{i\in\N}$ in $\USC(\Omega)$ for which there exists $K>0$ such that \eqref{Phi bounded below} is satisfied for each $\Phi_i$. If $\Phi_i\downarrow\Phi$, then
\[
\inf_{i\in\N} D^N(\Phi_i)= D^N(\Phi),\quad \forall N\in\N.
\]
\end{prop}

\begin{proof}
As the infimum of a sequence of upper semicontinuous functions satisfying \eqref{Phi bounded below}, $\Phi$ is again upper semicontinuous and satisfies \eqref{Phi bounded below}. It then follows from Proposition~\ref{prop:duality usc N} that 
\[
\inf_{i\in\N} D^N(\Phi_i)=\inf_{i\in\N} P^N(\Phi_i)=P^N(\Phi)=D^N(\Phi),
\]
where the second equality is due to Proposition~\ref{prop:P^N downward}. 
\end{proof}

The upward continuity of $D^N$, by contrast, is much more obscure. 
We need the following technical result, Lemma~\ref{lem:compactness for D}, to construct certain compactness for the space of semi-static strategies $(u,\Delta)$, which will facilitate the derivation of the upward continuity of $D^N$ in Proposition~\ref{prop:D^N upward} below. This lemma can be viewed as a generalization of \cite[Lemma 1.27]{Kellerer84} to the case of martingale optimal transport. The main idea involved is to extract additional compactness from the portfolio bound $N>0$ through Tychonoff's theorem.   

In Lemma~\ref{lem:compactness for D} below, let $\mathcal{B}(\R^t_+)$ be equipped with the topology of pointwise convergence. In addition, consider the product measure $\nu:=\mu_1\otimes\cdots\otimes\mu_T$ on $\Omega$, and denote by $L^1(\mu_t)$ (resp. $L^1(\nu)$) the set of $\mu_t$-integrable (resp. $\nu$-integrable) functions. %For any $u\in L^1(\mu)=\Pi_{t=1}^T L^1(\mu_t)$, recall the notation $\mu(u):=\sum_{t=1}^T\int_{\R_+}u_{t}(y)d\mu_{t}(y)$ used in \eqref{D}; in addition, we will write
Also recall $m(\mu_t)$, $t=1,...,T$, from \eqref{finite moment}. %:=\int_{\R_+}yd\mu_t(y)
\begin{lem}\label{lem:compactness for D}
Fix $N\in\N$ and $\Phi\in \cG(\Omega)$ that satisfies \eqref{Phi bounded below'} and ${D}^N(\Phi)<\infty$. For any $\delta > {D}^N(\Phi)$, define $\mathcal{L}(\Phi,\delta,N)$ as the collection of all pairs $(\Theta,\Delta)$, with
\begin{equation}\label{Theta}
\Theta:=\left\{(u^k_{1},...,u^k_{T},W^k)\right\}_{k\in\N}\in\left(\Pi_{t=1}^T L^1(\mu_t)\times L^1(\nu)\right)^\N  \;\;\text{ \rm and }\;\; \Delta\in\cH^N,
\end{equation}
 satisfying 
\begin{enumerate}[label={\rm (\roman*)},itemindent=.1in]
\item  For each $k\in\N$, $0\le u_{t}^k\le 2k$, $\forall t=1,\cdots,T$;
\item  $u^1_{t}\le u^2_{t}\le \cdots$, $\forall t=1,\cdots,T$;
\item  For each $k\in\N$, $\mu(u^k)\le \delta+(K+2N)(1+m(\mu_{1})+\cdots+m(\mu_{T}))$;
\item  For each $k\in\N$, $W^k\in L^1(\nu)$ with $0\le W^k\le \Lambda$, where $\Lambda\in L^1(\nu)$ is defined by
\[
\Lambda(x) := 2N(x_1+\cdots+x_T);  %2N(1+x_1+\cdots+x_T);
\]
moreover, $W^k=0$ on the set $\{x:\Lambda(x) <k\}$;
%moreover, $W^k\to 0$ as $k\to\infty$;
%\item {\color{red}$W^k\in L_+^0$ (replaced by $W^k\in L^1_+(\nu)$?)}, $\Delta\in\cH^N$, and {\color{red}$(\Delta\cdot x)_T-W^k\le k$(replaced by $(\Delta\cdot x)_T-W^k\to 0$ as $k\to\infty$ ?)};
\item  For each $k\in\N$, $\oplus u^k  \ge (\Phi+\Gamma)\wedge k+(\Delta\cdot x)_T-W^k$, where $\Gamma\in L^1(\nu)$ is defined by
\[
\Gamma(x) := (K+2N)(1+x_1+\cdots+x_T).
\]
Here, the constant $K>0$ in (iii) and (v) comes from \eqref{Phi bounded below'}.
\end{enumerate}
The set $\mathcal{L}(\Phi,\delta,N)$ is a nonempty compact subset of $\left(\Pi_{t=1}^T L^1(\mu_t)\times L^1(\nu)\right)^\N\times\Pi_{t=1}^{T-1}\mathcal{B}(\R^t_+)$, under the product of the weak topologies of the spaces $L^1(\mu_t)$, $L^1(\nu)$, and $\mathcal{B}(\R^t_+)$.
\end{lem}

\begin{proof} 
{\bf Step 1:} We show that $\mathcal{L}(\Phi,\delta,N)$ is {\it nonempty}. As $\delta > {D}^N(\Phi)$, there exist $u=(u_1,...,u_T)\in L^1(\mu)$ and $\Delta\in\mathcal{H}^N$ such that 
\[
\mu(u)\le\delta\ \ \ \hbox{and}\ \ \ \oplus u + (\Delta\cdot x)_T\ge \Phi.
\] 
As $\Phi$ satisfies \eqref{Phi bounded below'} and $|(\Delta\cdot x)_T|\le \Lambda(x)$, we have $\oplus u(x)\ge \Phi(x)-(\Delta\cdot x)_T\ge -\Gamma(x)$.
This implies that we can find constants $a_{1},a_2,...,a_{T}$ such that $\sum_{t=1}^T a_{t}=0$ and $a_{t}+u_{t}\ge -(K+2N)(1/T +x_{t})$ %2N(\frac{1}{T}+x_{t})
for all $t=1,...,T$. Now, define $\bar u_{t}:=a_{t}+u_{t}+(K+2N)(1/T+ x_{t})\ge 0$ for all $t=1,...,T$. Then, one can write
\[
\oplus \bar u \ge \Phi+ \Gamma + (\bar\Delta\cdot x)_T,\ \ \ \text{ \rm with } \bar\Delta:=-\Delta\in \cH^N.
\]
On the other hand, by the concavity of $x\mapsto x\wedge(2 k)$, 
\[
(\oplus \bar u)\wedge (2k)\ge \left((\Phi+ \Gamma)+ (\bar \Delta\cdot x)_T\right)\wedge (2k)\ge (\Phi+ \Gamma)\wedge k + (\bar\Delta\cdot x)_T\wedge k.
\]
Since $\bar u_t\ge 0$ for all $t=1,...,T$, it can be checked that $\oplus (\bar u\wedge (2k))\ge (\oplus \bar u)\wedge (2k)$. This, together with the previous inequality, gives
\begin{equation}\label{eqn:(v)}
\oplus (\bar u\wedge (2k))\ge (\Phi+ \Gamma)\wedge k + (\bar\Delta\cdot x)_T\wedge k.
\end{equation}
We claim that $u_t^k:=\bar u_{t}\wedge (2k)$, $W^k:=(\bar\Delta\cdot x)_T-(\bar\Delta\cdot x)_T\wedge k$, and $\bar\Delta$ form an element of $\mathcal{L}(\Phi,\delta,N)$. By construction, it is straightforward to verify conditions (i), (ii), and (v). Since $\oplus \bar u^k\le \oplus \bar u=\oplus u+\Gamma$, we have $\mu(\bar u^k)\le\delta+(K+2N)\left(1+m(\mu_{1})+\cdots+m(\mu_{T})\right)$, i.e. condition (iii) is satisfied. For each $k\in\N$, observe that $0\le W^k\le |(\bar\Delta\cdot x)_T|\le \Lambda(x)$. In particular, if $\Lambda(x)\le k$, then $|(\bar\Delta\cdot x)_T|\le k$ and thus $W^k=0$ by definition. This shows that condition (iv) is satisfied.
 
{\bf Step 2:} We prove that $\mathcal{L}(\Phi,\delta,N)$ is contained in a {\it weakly compact} space of functions. Observe that the following collections of functions
\begin{align*}
U(t,k)&:=\{u\in L^1(\mu_t):0\le u\le 2k\}\quad  t=1,...,T\ \hbox{and}\ k\in\N,\\
V&:=\{W\in L^1(\nu): 0\le W \le \Lambda\}
\end{align*}
are all uniformly integrable, and thus relatively weakly compact thanks to the Dunford-Pettis theorem. It follows that the countable product $(\Pi_{t,k} U(t,k))\times V^\N$ is also relatively weakly compact. 
On the other hand, for each $t=1,\cdots,T-1$, 
\[
F_t:=\{f:\R^t_+\to\R : |f|\le N\} = \Pi_{x\in\R^t_+} [-N,N]^x 
\]
is compact under the topology of pointwise convergence, as a consequence of Tychonoff's theorem. The space $F_t$ is therefore weakly compact, and this carries over to the product space $\cH^N = \Pi_{t}F_t$. We then conclude that $\Pi_{t,k} U(t,k)\times V^\N\times \cH^N$ is a weakly compact set containing $\mathcal{L}(\Phi,\delta,N)$. 

{\bf Step 3:} We prove that $\mathcal{L}(\Phi,\delta,N)$ is {\it strongly closed}. %We will show that it is strongly closed and convex, and therefore weakly closed. 
Take a sequence 
\[
\left\{\{(u^{k,m}_1,\cdots,u^{k,m}_T, W^{k,m})\}_{k\in\N},\Delta^m\right\}_{m\in\N}
\]
in $\mathcal{L}(\Phi,\delta,N)$ such that it converges to $(\{(u^k_1,\cdots,u^k_T,W^k)\}_{k\in\N},\Delta)$ in the strong sense. That is, $u^{k,m}_t\to u^k_t$ in $L^1(\mu_t)$, $W^{k,m}\to W^k$ in $L^1(\nu)$, and $\Delta^{m}\to\Delta$ pointwise in $\cH^N$. We intend to show that $(\{(u^k_1,\cdots,u^k_T,W^k)\}_{k\in\N},\Delta)$ also lies in $\mathcal{L}(\Phi,\delta,N)$.
%First notice that 
%For each $t=0,\cdots,T-1$, 
%\begin{align}
%\cI_t&:=\{f:\R^t_+\mapsto\R \mid |f|\le N\} = \Pi_{x\in\R^t_+} [-N,N]^x,
%\end{align}
%is compact under pointwise convergence, as a consequence of Tychonoff's theorem. Therefore, there exist a convergence subsequence of $\{\Delta^{\ell}\}_{\ell}$ with pointwise limit $\Delta\in\cH^N$. Therefore up to a subsequence, $(\Delta^{\ell}\cdot x)_T$ converges to $(\Delta\cdot x)_T$. 
%On the other hand, 

The convergence in $L^1(\mu_t)$ (resp. $L^1(\nu)$) implies the existence of a subsequence that converges $\mu_t$-a.e (resp. $\nu$-a.e.). 
%which converges almost everywhere, we  conclude that  
%$(\Delta^{\ell}\cdot x)_T-W^{k,\ell}$ converges to $Z$ $\nu$-a.s..
%Therefore, $W^{k,\ell}$ converges to the non-negative $W:=(\Delta\cdot x)_T-Z$ $\nu$-a.s.. In other words, $Z=(\Delta\cdot x)_T-W$ with $W\in L_+^1$, $\Delta\in\cH^N$, and $Z\le k$ up to a $\nu$-null set.
%On the other hand, since $u^{k,\ell}_t\to u^k_t$ in $L^1(\mu_t)$ implies that there exists a subsequence which converges almost everywhere, we conclude that 
Then, as $m\to\infty$, we conclude from $\oplus u^{k,m} \ge (\Phi+\Gamma)\wedge k+(\Delta^m\cdot x)_T-W^{k,m}$ that
\begin{equation}\label{eqn:property v}
\oplus u^k \ge (\Phi+\Gamma)\wedge k+(\Delta\cdot x)_T-W^k 
\end{equation}
holds outside a $\nu$-null set $\mathcal{N}$. %a set $N\subseteq \bigcup_{t=1}^TS_t^{-1}(N_t)$ where $N_t\subseteq\R_+$ is a $\mu_t$-null set.
We can then modify $(u^k_t)_{t=1}^T$ and $W^k$ on $\mathcal{N}$ such that \eqref{eqn:property v} holds everywhere, i.e. condition (vi) is satisfied. Also, we see from the convergence $u^{k,m}_t\to u^k_t$ and $\Delta^m\to\Delta$ that conditions (i), (ii), and (v) are satisfied, and Fatou's lemma implies the validity of (iii). From the convergence $W^{k,m}\to W^k$, we have $0\le W^k\le \Lambda$. Moreover, $W^k =0$ on $\{x:\Lambda(x)<k\}$ because $W^{k,m}=0$ on $\{x:\Lambda(x)<k\}$ for all $m\in\N$. This shows that condition (iv) is satisfied. We therefore conclude that $(\{(u^k_1,\cdots,u^k_T,W^k)\}_{k\in\N},\Delta)\in\mathcal{L}(\Phi,\delta,N)$, and thus $\mathcal{L}(\Phi,\delta,N)$ is closed under the strong topology. 

{\bf Step 4:} We prove the desired compactness of $\mathcal{L}(\Phi,\delta,N)$. Observe that $\mathcal{L}(\Phi,\delta,N)$ is convex. Since a strongly closed convex set is also weakly closed, and the weak topology of a product space coincides with the product of the weak topologies, we conclude that $\mathcal{L}(\Phi,\delta,N)$ is closed under the product of the weak topologies in the spaces $L^1(\mu_t)$, $L^1(\nu)$, and $\mathcal{B}(\R^t_+)$. It is therefore weakly compact in view of Step 2.
\end{proof}

\begin{rem}
While the motivation of Lemma~\ref{lem:compactness for D} is to construct some compactness for the space of semi-static strategies $(u,\Delta)$, we have to introduce the auxiliary random variable $W^k$ in \eqref{Theta} to ensure the convexity of $\mathcal{L}(\Phi,\delta,N)$, needed in the last step of the proof.  
\end{rem}

\begin{prop}\label{prop:D^N upward}
Consider $\{\Phi_i\}_{i\in\N}$ in $\cG(\Omega)$ for which there exists $K>0$ such that \eqref{Phi bounded below'} is satisfied for all $i\in\N$. If $\Phi_i\uparrow\Phi$,  then
\[
\sup_{i\in\N} {D}^N(\Phi_i) = {D}^N(\Phi),\quad \forall N\in\N.
\]
\end{prop}

\begin{proof}
Fix $N\in\N$. As $\Phi_i\uparrow\Phi$ clearly implies $\sup_{i\in\N}D^N(\Phi_i)\le D^N(\Phi)$, we focus on proving the ``$\ge$'' relation.  Assume $\sup_{i\in\N} {D}^N(\Phi)<\infty$, otherwise the proof would be trivial. Pick an arbitrary $\delta>\sup_{i\in\N} {D}^N(\Phi_i)$. By Lemma~\ref{lem:compactness for D}, $\{\mathcal{L}(\Phi_i,\delta,N)\}_{i\in\N}$ is a nonincreasing sequence of nonempty compact sets. We can therefore choose some $\left(\{(u^k_1,...,u^k_T,W^k)\}_{k\in\N},\Delta\right)\in\bigcap_{i\in\N}\mathcal{L}(\Phi_i,\delta,N_i)$. In view of conditions (i), (ii), and (iii) in Lemma~\ref{lem:compactness for D}, $u_t:=\lim_{k\to\infty}\uparrow u^k_t\in L^1(\mu_t)$ is well-defined, and $u=(u_1,...,u_T)$ satisfies
\begin{equation}\label{bdd by delta}
\mu(u)\le\delta+(K+2N)(1+m(\mu_{1})+\cdots+m(\mu_{T})).
\end{equation}
Moreover, condition (v) in Lemma~\ref{lem:compactness for D} implies that for each $k$ and $i$,
\[
\oplus u^k  \ge (\Phi_i+\Gamma)\wedge k+(\Delta\cdot x)_T-W^k.
\]
Recall from condition (iv) in Lemma~\ref{lem:compactness for D} that $W^k=0$ on $\{x:\Lambda(x)<k\}$. This in particular implies $W^k(x)\to 0$ for all $x\in\Omega$ as $k\to\infty$. Therefore, by taking $k\to\infty$ in the previous inequality, we get $\oplus u  \ge \Phi_i+\Gamma+(\Delta\cdot x)_T$. As $i\to\infty$, this yields 
\begin{equation}\label{no i}
\oplus u \ge \Phi+\Gamma+ (\Delta\cdot x)_T.
\end{equation}
Now, define $\bar u_t := u_t - (K+2N)(1/T+ x_t)$ for all $t=1,\cdots,T$. By \eqref{bdd by delta} and \eqref{no i},
\begin{align*}
\mu(\bar u) &= \mu(u) -(K+2N)(1+m(\mu_{1})+\cdots+m(\mu_{T}))\le \delta,\\
\oplus \bar u &= \oplus u -\Gamma \ge \Phi+ (\Delta\cdot x)_T.
\end{align*}
This readily implies ${D}^N(\Phi)\le\delta$. With $\delta>\sup_{i\in\N}{D}^N(\Phi_i)$ arbitrarily chosen, we conclude $\sup_{i\in\N}{D}^N(\Phi_i)\ge D^N(\Phi)$.
\end{proof}

\subsection{Complete Duality}

\begin{thm}\label{thm:P^N=D^N}
For any $\Phi\in\operatorname{USA}(\Omega)$ that satisfies \eqref{Phi bounded},
\begin{equation}\label{P^N=D^N}
{D}^N(\Phi)=P^N(\Phi),\quad \forall N\in\N. 
\end{equation}
Moreover, there exists an optimizer $(u,\Delta)\in L^1(\mu)\times\cH^N$ for $D^N(\Phi)$ whenever $D^N(\Phi)<\infty$.
\end{thm} 

\begin{proof}
Fix $N\in\N$. Define $\zeta^K(x):=K(1+x_1+...+x_T)$, with $K>0$ specified in \eqref{Phi bounded}. Consider the functionals $\bar P^N$ and $\bar D^N$ defined by
\[
\bar P^N(\varphi) := P^N(-\zeta^K\vee(\varphi\wedge\zeta^K))\quad \hbox{and}\quad \bar D^N(\varphi) := D^N(-\zeta^K\vee(\varphi\wedge\zeta^K)),\quad \hbox{for}\ \varphi\in\cG(\Omega).
\] 
In view of Propositions~\ref{prop:P^N upward} and \ref{prop:P^N downward} (resp. Propositions~\ref{prop:D^N downward} and \ref{prop:D^N upward}), $\bar P^N$ (resp. $\bar D^N$) is a Choquet capacity associated with $\USC(\Omega)$; recall Definition~\ref{defn:choquet-capcity}. Moreover, thanks to Proposition~\ref{prop:duality usc N}, $\bar D^N(\varphi) = \bar P^N(\varphi)$ for all $\varphi\in\USC(\Omega)$. We then conclude from Lemma~\ref{lem:capacity} that $\bar D^N(\varphi) = \bar P^N(\varphi)$ for all $\varphi\in\USA(\Omega)$. That is to say, $D^N(\varphi)=P^N(\varphi) $ for all $\varphi\in\USA(\Omega)$ satisfying $|\varphi|\le\zeta^K$, or \eqref{Phi bounded}. 

It remains to prove the existence of an optimizer for $D^N(\Phi)$. If $D^N(\Phi)<\infty$, take a real sequence $\{\delta_i\}$ such that $\delta_i\downarrow D^N(\Phi)$. By Lemma~\ref{lem:compactness for D}, $\{\mathcal{L}(\Phi,\delta_i,N)\}_{i\in\N}$ is a nonincreasing sequence of nonempty compact sets. We can therefore choose some $\left(\{(u^k_1,...,u^k_T,W^k)\}_{k\in\N},\Delta\right)\in\bigcap_{i\in\N}\mathcal{L}(\Phi,\delta_i,N)$. 
In view of conditions (i), (ii), and (iii) in Lemma~\ref{lem:compactness for D}, $u_t:=\lim_{k\to\infty}\uparrow u^k_t\in L^1(\mu_t)$ is well-defined, and $u=(u_1,...,u_T)$ satisfies
\begin{equation}\label{bdd by delta'}
\mu(u)\le D^N(\Phi)+(K+2N)(1+m(\mu_{1})+\cdots+m(\mu_{T})).
\end{equation}
Moreover, condition (v) in Lemma~\ref{lem:compactness for D} implies that for each $k$ and $i$,
\[
\oplus u^k  \ge (\Phi+\Gamma)\wedge k+(\Delta\cdot x)_T-W^k.
\]
As shown in the proof of Proposition~\ref{prop:D^N upward}, $W^k(x)\to 0$ for all $x\in\Omega$ as $k\to\infty$. Thus, by taking $k\to\infty$ in the previous inequality, we get $\oplus u  \ge \Phi+\Gamma+(\Delta\cdot x)_T$. Now, define $\bar u_t := u_t - (K+2N)(1/T+ x_t)$ for all $t=1,\cdots,T$. Then, $\oplus \bar u = \oplus u -\Gamma \ge \Phi+ (\Delta\cdot x)_T$. Moreover, by \eqref{bdd by delta'},
\begin{align*}
\mu(\bar u) &= \mu(u) -(K+2N)(1+m(\mu_{1})+\cdots+m(\mu_{T}))\le D^N(\Phi).
\end{align*}
This implies that, $(\bar u,-\Delta)\in L^1(\mu)\times\cH^N$ is an optimizer of ${D}^N(\Phi)$. 
\end{proof}

\begin{rem}\label{rem:cannot apply}
When we view $D$ and $P$, defined in \eqref{D} and \eqref{primal}, as functionals, arguments similar to (and simpler than) those in Sections~\ref{subsec:P^N} and \ref{subsec:D^N} yield the upward and downward continuity of $P$, as well as the downward continuity of $D$. However, the upward continuity of $D$ is obscure. Without the portfolio bound $N>0$, it is unclear how the space of semi-static strategies $(u,\Delta)\in L^1(\mu)\times \cH$ can be made compact under any topology, so that the upward continuity does not follow from the arguments in Proposition~\ref{prop:D^N upward}. %This prevents a direct application of Choquet's capacitability theorem to the unconstrained duality $D(\Phi)=P(\Phi)$ in Proposition~\ref{prop:duality usc}. 

In fact, since $D(\Phi)\neq P(\Phi)$ for some Borel measurable $\Phi$ (as shown in \cite[Example 3.1]{BNT17}), the upward continuity of $D$ must not hold. Otherwise, we could apply Choquet's capacitability theorem directly to the classical duality $D(\Phi)=P(\Phi)$ in Proposition~\ref{prop:duality usc}, extending it from upper semicontinuous $\Phi$ to upper semi-analytic ones (which include Borel measurable ones).
\end{rem}

%The next example demonstrates explicitly that despite $D(\Phi)>P(\Phi)$, we have $D^N(\Phi) = P^N(\Phi)$ for all $N\in\N$.  

\begin{rem}
Recall Example~\ref{eg:eg 8.1 from BNT}, where $0=P(\Phi)<D(\Phi)=1$. We will show that $P^N(\Phi) = D^N(\Phi)$ for all $N\in\N$. Fix $N\in\N$. Recall that $(u^*_1,u^*_2,\Delta^*_1)\equiv (1,0,0)$ is an optimizer of $D(\Phi)$. As $\Delta^*_1\in\cH^N$, $(u^*_1,u^*_2,\Delta^*_1)$ is also an optimizer of $D^N(\Phi)$, and thus $D^N(\Phi) = D(\Phi)=1$. On the other hand, consider $\{\Q_M\}_{M\in\N}$ in $\Pi$ constructed in \eqref{density}. By \eqref{A^N_2}, $\bA^{N}_2(\Q_M)= N\bA^{1}_2(\Q_M)  = \frac{N}{4M}$. It follows that
\[
P^N(\Phi) = \sup_{\Q\in\Pi} \left\{\E_{\Q}[\Phi]-\bA^N_2(\Q)\right\} \ge \lim_{M\to\infty} \left\{\E_{\Q_M}[\Phi]-\bA^N_2(\Q_M)\right\}=1.
\]
As $\Phi\le 1$ already implies $P^N(\Phi)\le 1$, we conclude $P^N(\Phi)= 1=D^N(\Phi)$. 
\end{rem}

\begin{rem}
$P^N(\Phi)$ in general does not admit an optimizer, unless $\Phi$ is upper semicontinuous. To illustrate, in Example~\ref{eg:eg 8.1 from BNT}, suppose that there exists $\Q_*\in\Pi$ such that $\E_{\Q_*}[\Phi] - \bA^N_2(\Q_*) = P^N(\Phi)=1$ for some $N\in\N$. Then, 
$
0\le \bA^N_2(\Q_*) = \E_{\Q_*}[\Phi]-1\le 0,
$    
which yields $\bA^N_2(\Q_*)=0$. By Proposition~\ref{prop:cM equivalence}, $\Q_*$ must belong to $\cM$ and thus coincide with $\P_0$. This, however, entails $\E_{\Q_*}[\Phi]-\bA^N_2(\Q_*)= 0$, a contradiction.   
\end{rem}

%%%%%%%%%%%%%%%%%%%%%%%%%%%%%%%%%%%%%%
%%%%%%%%%%%%%%%%%%%%%%%%%%%%%%%%%%%%%%

\section{Derivation of Theorem \ref{thm:main}}\label{pf.thm.main}
This section is devoted to proving Theorem \ref{thm:main}. %First, we define $D^\infty(\Phi)$ as in \eqref{D}, with $\Delta\in\cH$ replaced by $\Delta\in \cH^\infty$, where
To connect the portfolio-constrained duality \eqref{P^N=D^N} to the desired (unconstrained) duality \eqref{general duality}, it is natural to relax the constraint $N>0$ by taking $N\to\infty$, leading to the next result. 
Recall $D^\infty(\Phi)$ defined above \eqref{H^infty} and $\widetilde P(\Phi)$ defined in \eqref{general duality}.

\begin{prop}\label{prop: D.infty}
For any $\Phi\in\USA(\Omega)$ that satisfies \eqref{Phi bounded}, $D^\infty(\Phi)=\widetilde P(\Phi)$. 
\end{prop}
\begin{proof}
First, we show that $D^\infty(\Phi)\ge \widetilde P(\Phi)$. Fix $\{\Q_N\}_{N\in\N}$ in $\Pi$ such that $\Q_N\overset{\rho}{\to}\cM$ (or equivalently, $\bA^1_T(\Q_N)\to 0$). We can choose some  nonnegative function $h$ such that $h(N)\to\infty$ and $h(N)\bA^1_T(\Q_N)\to 0$ (for instance, take $h(N):= 1/\sqrt{\bA^1_T(\Q_N)}$). For each $N\in\N$, there exist $(u,\Delta)\in L^1(\mu)\times\cH^{h(N)}$ with $\mu(u)< D^{h(N)}(\Phi)+1/N$ such that $\oplus u+(\Delta\cdot S)_T\ge \Phi$. If follows that 
\[
D^{h(N)}(\Phi)+1/N+ \bA^{h(N)}_T(\Q_N)\ge \mu(u) +\E_{\Q_N}[(\Delta\cdot S)_T] \ge \E_{\Q_N}[\Phi], 
\]
where the first inequality follows from the definition of $\bA^{h(N)}_T(\Q_N)$ in \eqref{def: A_t}. As $N\to\infty$ in the above inequality, since $\bA^{h(N)}_T(\Q_N) = h(N) \bA^1_T(\Q_N) \to 0$ and $D^{h(N)}(\Phi)\to D^\infty(\Phi)$ by definition, we get $D^\infty(\Phi)\ge \limsup_{N\to\infty} \E_{\Q_N}[\Phi]$. With $\{\Q_N\}_{N\in\N}$ arbitrarily chosen, we obtain $D^\infty(\Phi) \ge  \widetilde P(\Phi)$.

On the other hand, for any $N\in\N$, by the definition of $P^N(\Phi)$, we can take $\Q_N\in\Pi$ such that 
\begin{equation}\label{crucial estimate}
P^N(\Phi)\ge \E_{\Q_N}[\Phi]- \bA^{N}_T(\Q_N)> P^{N}(\Phi) - 1/N.  %\ge P^\infty(\Phi)-1/N.
\end{equation}
This, together with $\bA^{N}_T(\Q_N) = N \bA^1_T(\Q_N)$, shows that %there exists $C>0$ such that 
\begin{equation}\label{A_N to 0}
\bA^1_T(\Q_N)< \frac{\E_{\Q_N}[\Phi]-P^N(\Phi)+1/N}{N} \le \frac{C}{N},\quad \forall N\in\N.  
\end{equation}
Here, the constant $C>0$ can be chosen to be independent of $N$, thanks to \eqref{Phi bounded} and \eqref{finite moment}. This in particular implies $\bA^1_T(\Q_N)\to 0$. In view of \eqref{crucial estimate}, this yields 
\begin{equation}\label{P^N<limsup}
\lim_{N\to\infty} P^{N}(\Phi) = \lim_{N\to\infty}\{\E_{\Q_N}[\Phi]- \bA^{N}_T(\Q_N)\} \leq \limsup_{N\to\infty}\E_{\Q_N}[\Phi]\le \widetilde P(\Phi). 
\end{equation}
Finally, by taking $N\to\infty$ in the constrained duality \eqref{P^N=D^N} and using the above inequality, we obtain 
$D^\infty(\Phi) = \lim_{N\to\infty} P^N(\Phi)\le \widetilde P(\Phi)$.  
\end{proof}

In view of Proposition~\ref{prop: D.infty}, to obtain the desired duality \eqref{general duality}, it remains to show $D^{\infty}(\Phi) = D(\Phi)$ for all $\Phi\in\USA(\Omega)$ satisfying \eqref{Phi bounded}. That is, restricting to bounded trading strategies does not increase the cost of model-free superhedging. To this end, we need the following technical result. 

\begin{lem}\label{phi.eps approx}
Given $\Phi\in\cG(\Omega)$ that satisfies \eqref{Phi bounded}, we define $\Phi_n\in\cG(\Omega)$, for each $n\in\N$, by
\begin{equation}\label{Phi_n}
\Phi_{n}(x_1,...,x_T) := \Phi(x_1,...,x_T) 1_{\{x_1 \leq n, ..., x_T \leq n\}}(x_1,...,x_T),\quad \forall x=(x_1,...,x_T)\in\Omega. 
\end{equation}
For any $\eps>0$, there exists $n\in\N$ large enough such that 
\begin{equation}\label{Phi-Phi_n}
|\E_\Q[\Phi] - \E_\Q[\Phi_{n}] |< \eps,\quad\forall \Q\in\Pi.
\end{equation}
\end{lem}

\begin{proof}
Fix $\eps>0$. Let $\delta:= \frac{\eps}{K(T+T^2)}$. Thanks to \eqref{finite moment}, we can take $n\in\N$ large enough such that  
\begin{equation}\label{<eps}
%n \mu_t(x_t>n) \leq  
\mu_t((n,\infty))<\delta \quad \hbox{and}\quad \int_{\{y> n\}} y d\mu_t(y) <\delta,\quad\forall t=1,...,T.
%\E_{\mu_i}[x_i 1_{x_i>n}] < \epsilon \ \ \ \ \text{for all} \ \ 1 \leq i \leq T.
\end{equation}
For simplicity, we will write $\A = \{x\in\Omega: x_1 \leq n,..., x_T \leq n\}$. Observe that
\begin{equation}\label{A^c 1}
\A^c \subseteq  \bigcup_{t\in\{1,...,T\}}\{x\in\Omega: x_t > n\}. %\cup ... \cup \{x\in\R^T_+: x_T > n\}.
\end{equation}
Moreover, for each fixed $t=1,...,T$, 
\begin{align}\label{A^c 2}
\A^c =  \{x\in\Omega: x_t > n\} \cup \bigcup_{i\in\{1,...,T\}\setminus\{t\}}\{x\in\Omega: x_t \leq n\ \hbox{and}\ x_i>n\}. 
\end{align}
Now, for any $\Q\in\Pi$,  by \eqref{Phi bounded}, 
\begin{equation}\label{dif:phi-phi.eps}
\begin{split}
|\E_{\Q} &\left[\Phi \right] - \E_{\Q}\left[\Phi_{n} \right]| \le \E_{\Q}\left[|\Phi|1_{\A^c}\right]\le K \bigg(\E_{\Q}[1_{\A^c}] +  \sum_{t=1,...,T} \E_{\Q}[x_t1_{\A^c}(x) ] \bigg).
\end{split}
\end{equation}
The first inequality above requires the linearity of outer expectations; recall from Section~\ref{subsec:notation} that $\E_\Q[\cdot]$ denotes an outer expectation if the integrad need not be Borel measurable. While the linearity of outer expectations does not hold in general, it holds specifically here thanks to the definition of $\Phi_n$. Indeed, by \cite[Lemma 6.3]{Kosorok-book-08}, there exists $\Phi^*\in\cB(\Omega)$, a minimal Borel measurable majorant of $\Phi$, such that $\E_\Q[\Phi^*] = \E_\Q[\Phi]$ and $\E_\Q[\Phi^* 1_B] = \E_\Q[\Phi 1_B]$ for any Borel subset $B$ of $\Omega$. It follows that
\begin{equation*}
\E_{\Q} \left[\Phi \right] - \E_{\Q}\left[\Phi_{n} \right] = \E_{\Q} \left[\Phi^* \right] - \E_{\Q}\left[\Phi^*1_\A \right]= \E_{\Q}\left[\Phi^*1_{\A^c}\right]=\E_{\Q}\left[\Phi1_{\A^c}\right],
\end{equation*}
where the second equality follows from the linearity of standard expectations, as $\Phi^*$ and $\Phi^*1_\A$ are both Borel measurable.
 
Thanks to \eqref{A^c 1},
\[
\E_{\Q}[1_{\A^c}(x)] \le \sum_{t=1,...,T}\E_{\Q}[1_{\{x_t >n\}}(x)] = \sum_{t=1,...,T}\mu_t((n,\infty))<T\delta,
\]
where the last inequality follows from \eqref{<eps}. On the other hand, for any $t=1,...,T$, \eqref{A^c 2} implies
\begin{align*}
\E_{\Q}[x_t1_{\A^c}(x)] &= \E_{\Q}[x_t 1_{\{x_t > n\}}(x)] +\sum_{i\in\{1,...,T\}\setminus\{t\}} \E_{\Q}[x_t 1_{\{x_t \leq n,\ x_i>n\}}(x) ]\\
&\le  \E_{\Q}[x_t 1_{\{x_t > n\}}(x)] +\sum_{i\in\{1,...,T\}\setminus\{t\}} \E_{\Q}[x_i 1_{\{x_i>n\}}(x) ]\\
&= \sum_{i=1,...,T}  \E_{\Q}[x_i 1_{\{x_i > n\}}(x)] = \sum_{i=1,...,T} \int_{\{y> n\}} y d\mu_i(y) < T\delta, 
\end{align*}
where the last inequality follows from \eqref{<eps}. Hence, we conclude from \eqref{dif:phi-phi.eps} that $|\E_{\Q} \left[\Phi(x) \right] - \E_{\Q}\left[\Phi_{n}(x) \right]| \le K (T+T^2) \delta=\eps$, as desired. 
\end{proof}

\begin{coro}\label{coro:Phi-Phi_n}
If $D^\infty(\Phi)=D(\Phi)$ for all bounded $\Phi\in\USA(\Omega)$, then the same equality holds for all $\Phi\in\USA(\Omega)$ satisfying \eqref{Phi bounded}.
\end{coro}

\begin{proof}
First, we show that $D^\infty(\Phi)=D(\Phi)$ for all nonnegative $\Phi\in\USA(\Omega)$ satisfying \eqref{Phi bounded}. Given $\Phi\in\USA(\Omega)$ that is nonnegative and satisfies \eqref{Phi bounded}, consider $\Phi_n$, $n\in\N$, defined in \eqref{Phi_n}. As a product of $\Phi\in\USA(\Omega)$ and a nonnegative Borel measurable function, $\Phi_n$ also belongs to $\USA(\Omega)$, thanks to \cite[Lemma 7.30]{Bertsekas-Shreve-book}. In view of the estimate \eqref{Phi-Phi_n} and the definition of $\widetilde P$ in \eqref{general duality}, we deduce from Proposition~\ref{prop: D.infty} that $D^\infty(\Phi_n)=\widetilde P(\Phi_n)\to\widetilde P(\Phi)=D^\infty(\Phi)$. Now, note that every $\Phi_n$ is bounded, thanks to the fact that $\Phi$ satisfies \eqref{Phi bounded}. As the boundedness of $\Phi_n\in\USA(\Omega)$ implies $D^\infty(\Phi_n)=D(\Phi_n)$ for all $n\in\N$, we have
\[
D^\infty(\Phi)=\lim_{n\to\infty} D^\infty(\Phi_n)=\lim_{n\to\infty} D(\Phi_n)\le D(\Phi),
\]
where the inequality stems from $\Phi_n\uparrow\Phi$, thanks to the fact that $\Phi$ is nonnegative. Since $D^\infty(\Phi)\ge D(\Phi)$ by definition, we conclude that $D^\infty(\Phi)=D(\Phi)$. 

Now, take an arbitrary $\Phi\in\USA(\Omega)$ that satisfies \eqref{Phi bounded} (which need not be nonnegative). Consider $v=(v_1,...,v_T)\in L^1(\mu)$ defined by $v_t(y) := K(\frac1T+y)$ for $t=1,...,T$, where $K>0$ is taken from \eqref{Phi bounded}. As $\Phi$ satisfies \eqref{Phi bounded}, $\Phi+\oplus v$ is nonnegative. Indeed, $(\Phi+\oplus v)(x)\ge -K(1+x_1+...+x_T)+\sum_{t=1}^T K(\frac1T+x_t)=0$, for all $x\in\Omega$. Moreover, $\Phi+\oplus v$ again satisfies \eqref{Phi bounded}, with a possibly larger $K> 0$. Hence, we have 
\begin{equation}\label{11}
D^\infty(\Phi+\oplus v)=D(\Phi+\oplus v). 
\end{equation}
Note that $\E_\Q[\Phi+\oplus v] = \E_\Q[\Phi] + \mu(v)$ for all $\Q\in\Pi$. This, together with Proposition~\ref{prop: D.infty}, implies 
\begin{align}\label{22}
D^\infty(\Phi+\oplus v)=\widetilde P(\Phi+\oplus v)= \widetilde P(\Phi)+\mu(v) = D^\infty(\Phi)+\mu(v).
\end{align} 
On the other hand, by definition
\begin{align}
D(\Phi+\oplus v) &= \inf\{\mu(u): u\in L^1(\mu)\ \hbox{satisfying}\ \exists \Delta\in\cH\ \hbox{s.t.}\ \oplus u + (\Delta\cdot S)_T\ge \Phi+\oplus v\ \hbox{on}\ \Omega\}\nonumber\\
&= \inf\{\mu(u): u\in L^1(\mu)\ \hbox{satisfying}\ \exists \Delta\in\cH\ \hbox{s.t.}\ \oplus (u-v) + (\Delta\cdot S)_T\ge \Phi\ \hbox{on}\ \Omega\}\nonumber\\
&= \inf\{\mu(\tilde u)+\mu(v): \tilde u\in L^1(\mu)\ \hbox{satisfying}\ \exists \Delta\in\cH\ \hbox{s.t.}\ \oplus \tilde u + (\Delta\cdot S)_T\ge \Phi\ \hbox{on}\ \Omega\}\nonumber\\
&= D(\Phi)+\mu(v).\label{33}
\end{align}
On the strength of \eqref{22} and \eqref{33}, \eqref{11} yields $D^\infty(\Phi)=D(\Phi)$.
 \end{proof}

Now, we are ready to establish $D^{\infty}(\Phi) = D(\Phi)$ for all upper semi-analytic $\Phi$ satisfying \eqref{Phi bounded}. 

\begin{prop}\label{thm: sec.main}
For any $\Phi\in\USA(\Omega)$ that satisfies \eqref{Phi bounded}, $D^{\infty}(\Phi) = D(\Phi)$.
\end{prop}

\begin{proof}
First, by Corollary~\ref{coro:Phi-Phi_n}, we can assume without loss of generality that $\Phi\in\USA(\Omega)$ is bounded. We take $C>0$ such that $|\Phi|\le C$ on $\Omega$.
 
As $D^{\infty}(\Phi) \geq D(\Phi)$ by definition, we focus on proving the opposite inequality. %Assume $D(\Phi)<\infty$, otherwise the proof would be trivial. 
Fix $\delta>0$. There exist $u=(u_1,...,u_T)\in L^1(\mu)$ and $\Delta\in \cH$ such that  
\begin{align}\label{u+(d.s)>phi'}
\mu(u)<D(\Phi)+\delta/2\quad \hbox{and}\quad \oplus u(x)+ (\Delta\cdot x)_T &\geq \Phi(x)\quad \forall x\in\Omega.  
\end{align}

{\bf Step 1:} We replace $u\in L^1(\mu)$ by nonnegative functions. 
By the Vitali-Carath\'{e}odory theorem, there exists $v=(v_1,...,v_T)\in L^1(\mu)$, with $u_t\le v_t $ and $v_t$ bounded from below for all $t=1,...,T$, such that $\mu(u)\le\mu(v)\le \mu(u)+\delta/2$. Take $\ell>0$ large enough such that $v_t\ge -\ell$ for all $t=1,...,T$. By setting $\bar v_t:= v_t+\ell\ge 0$, we deduce from \eqref{u+(d.s)>phi'} that  
\begin{align}\label{u+(d.s)>phi}
\mu(v)<D(\Phi)+\delta\quad \hbox{and}\quad \oplus \bar v(x)+ (\Delta\cdot x)_T &\geq \Phi(x)+T\ell\quad \forall x\in\Omega.  
\end{align}

{\bf Step 2:} We construct a bounded trading strategy $\bar \Delta\in \cH^\infty$ and replace \eqref{u+(d.s)>phi} by a superhedging relation involving $\bar\Delta$.  
Fix arbitrary $\eps_1,\eps_2,...,\eps_{T-1} >0$ that are sufficiently small. 
As $\bar v_1$ is $\mu_1$-integrable, by \cite[Problem 14, p.63]{MR1013117}, there exists $M_1\in\cB(\R_+)$ such that $\mu_1(\R_+\setminus M_1)<\eps_1$ and $\bar v_1$ is bounded on $M_1$. We can assume without loss of generality that $M_1$ contains $\{0\}$. Indeed, if $\mu_1(\{0\})=0$, adding $\{0\}$ to $M_1$ does not change the above statement; if $\mu_1(\{0\})>0$, then $M_1$ has to contain $\{0\}$ as long as $\eps_1<\mu_1(\{0\})$. For any $m_1>1$, define 
\[
\widetilde M_1:= M_1 \cap (\{0\}\cup(1/m_1,m_1)).
\]
Note that $\mu_1(\widetilde M_1)\uparrow \mu_1(M_1)$ as $m_1\to\infty$. Now, we claim that 
\[
\Delta_1\ \hbox{is bounded on}\ \widetilde M_1,\quad \forall m_1>1.
\]
By contradiction, suppose that there exist $\{x_1^n\}_{n\in\N}$ in $\widetilde M_1$ such that $\Delta_1(x_1^n) \to\infty$ or $-\infty$. By taking $x_1 = x_1^n$ and $x_2=x_3=...=x_T\in\R_+$ in the second part of \eqref{u+(d.s)>phi} and using the fact $|\Phi|\le C$, we get 
\begin{align}\label{dududu}
\bar v_1(x^n_1)+ \bar v_2(x_2)+\dots &+ \bar v_T(x_2) + \Delta_1(x^n_1)(x_2-x^n_1) \ge -C+T\ell.
\end{align}
For the case $\Delta_1(x^n_1) \to \infty$ (resp. $\Delta_1(x^n_1) \to -\infty$), we take $x_2=\frac{1}{2m_1}$ (resp. $x_2=m_1+1$) in \eqref{dududu}. As $n\to\infty$, by the boundedness of $\bar v_1$ on $\widetilde M_1$, the left hand side of  \eqref{dududu} tends to $-\infty$, a contradiction. 
Similarly to the above, by \cite[Problem 14, p.63]{MR1013117}, there exists $M_2\in\cB(\R_+)$, containing $\{0\}$, such that $\mu_2(\R_+\setminus M_2)<\eps_2$ and $\bar v_2$ is bounded on $M_2$. For any $m_2>1$, define 
\[
\widetilde M_2:= M_2 \cap (\{0\}\cup(1/m_2,m_2)),
\]
and note that $\mu_2(\widetilde M_2)\uparrow \mu_2(M_2)$ as $m_2\to\infty$. We claim that 
\[
\Delta_2\ \hbox{is bounded on}\ \widetilde M_1\times\widetilde M_2,\quad \forall m_1, m_2>1.
\]
By contradiction, suppose that there exist $\{(x_1^n,x_2^n)\}_{n\in\N}$ in $\widetilde M_1 \times \widetilde M_2$ such that $\Delta_2(x_1^n,x_2^n) \to \infty$ {or} $-\infty$. By taking $(x_1,x_2) = (x_1^n,x_2^n)$ and $x_3=x_4=...=x_T\in\R_+$ in the second part of \eqref{u+(d.s)>phi}  and using the fact $|\Phi|\le C$, we get 
\begin{align}\label{dudu}
\bar v_1(x^n_1)+ \bar v_2(x^n_2)+\bar v_3(x_3)+ ...+ \bar v_T(x_3) + \Delta_1(x^n_1)(x^n_2-x^n_1)+ \Delta_2(x^n_1,x^n_2 )&(x_3-x^n_2)\nonumber\\ 
&\ge -C+T\ell. 
\end{align}
For the case $\Delta_2(x_1^n,x_2^n) \to\infty$ (resp. $\Delta_2(x_1^n,x_2^n) \to-\infty$), we take $x_3=\frac{1}{2m_2}$ (resp. $x_3=m_2+1$) in \eqref{dudu}. As $n\to\infty$, by the boundedness of $\bar v_1$ (on $\widetilde M_1$), $\bar v_2$ (on $\widetilde M_2$), and $\Delta_1$ (on $\widetilde M_1$), the left hand side of  \eqref{dudu} tends to $-\infty$, a contradiction. %For the case $\Delta_2(x_1^n,x_2^n) \to-\infty$, we take $x_3=m_2+1$. As $n\to\infty$, by the boundedness of $u_1$ (on $\widetilde M_1$), $u_2$ (on $\widetilde M_2$), and $\Delta_1$ (on $\widetilde M_1$), the left hand side of  \eqref{dudu} approaches $-\infty$, a contradiction.
By repeating the same argument for all $t=3,4,...,T-1$, we obtain $\{M_t\}_{t=1}^{T-1}$ in $\cB(\R_+)$ such that for each $t=1,...,T-1$, 
\begin{itemize}
\item [(i)] $\mu_t(M^c_t)=\mu_t(\R_+\setminus M_t)<\eps_t$;
\item [(ii)] $\mu_t(\widetilde M_t)\uparrow \mu_t(M_t)$ as $m_t\to\infty$, where $\widetilde M_t := M_t\cap(\{0\}\cup(1/m_t, m_t))$ for $m_t>1$;
\item [(iii)] $\Delta_t(x_1,x_2,\dots,x_t)\ \hbox{is bounded on}\ \widetilde M_1 \times \widetilde M_2 \times ...\times \widetilde M_t$. 
\end{itemize}
We also consider
\begin{equation}\label{a_t}
a_t := \sup_{\widetilde M_1 \times \widetilde M_2 \times ...\times \widetilde M_t} |\Delta_t(x_1,x_2,\dots,x_t)| <\infty,
\end{equation}
for all $t=1,...,T-1$, which will be used in Step 3 of the proof. 

Now, let us define the bounded strategy $\bar\Delta=\{\bar\Delta_t\}_{t=1}^{T-1}\in\cH^\infty$  by
\begin{equation*}
    \bar \Delta_t(x_1, x_2, \dots, x_t):=  \Delta_t(x_1, x_2, \dots, x_t) 1_{\widetilde M_1 \times... \times \widetilde M_t}(x_1,x_2, ..., x_t),\quad\forall t=1,...,T-1.
\end{equation*}
Also, for any $x=(x_1,...,x_T)\in\Omega$, we introduce 
\begin{align}\label{barPhi}
{\bar\Phi}(x) := (\Phi(x)+T\ell) 1_{ \widetilde M_1 \times...\times  \widetilde M_{T-1}}(x) 
+\sum_{t=2}^{T-1} (\Delta\cdot S)_{t} 1_{ \widetilde M_1 \times ...\times  \widetilde M_{t-1} \times  \widetilde{M}^c_t}(x_1,x_2,...,x_t).
\end{align}
We claim that 
\begin{align}\label{ine:dot.bar.phi}
\oplus \bar v(x) + (\bar \Delta\cdot x)_T \geq {\bar\Phi}(x),\quad \forall x\in\Omega.
\end{align} 
Indeed, for any $x\in\Omega$ such that $x_t \in \widetilde M_t$ for all $t=1,..., T-1$, the above inequality simply reduces to the second part of \eqref{u+(d.s)>phi}. For any $x\in\Omega$ such that $x_t \notin \widetilde M_t$ for some $t=1,...,T-1$, consider $t^* := \inf\{t\in\{1,2,...,T-1\}: x_t \notin \widetilde M_t\}$. Observe that 
\begin{align*}
\oplus \bar v(x) + (\bar \Delta\cdot x)_T&=  \oplus \bar v(x) + \Delta_1(x_1)(x_2-x_1)+\dots + \Delta_{t^*-1}(x_1,x_2,\dots,x_{t^*-1})(x_{t^*}-x_{t^*-1})\\
&=\oplus \bar v(x) + (\Delta\cdot x)_{t^*}\ge (\Delta\cdot x)_{t^*} = {\bar\Phi}(x),
\end{align*}
where the inequality follows from $\bar v_t\ge 0$, and the last equality is deduced from the definitions of $\bar\Phi$ and $t^*$. We therefore conclude that \eqref{ine:dot.bar.phi} holds.

{\bf Step 3:} We show that for any $\eps>0$, $\{\widetilde M_t\}_{t=1}^{T-1}$ can be constructed appropriately so that $\E_\Q[{\bar\Phi}]\ge \E_{\Q}[\Phi+T\ell]-\eps$ for all $\Q\in\Pi$. 
For any $\Q\in\Pi$, by \eqref{barPhi},  
\begin{align}\label{d.infty(dot.bar.phi)}
\E_\Q[{\bar\Phi}]&= \E_{\Q} \left[(\Phi+T\ell) 1_{\widetilde M_1 \times ... \times \widetilde M_{T-1}}\right]+ \sum_{t=2}^{T-1}\E_{\Q}\left[ (\Delta\cdot S)_{t} 1_{\widetilde M_1 \times ... \times \widetilde M_{t-1} \times \widetilde M^c_t} \right]. 
\end{align}
Note that 
\begin{align}\label{first term}
\begin{split}
\E_{\Q}&\left[ (\Phi+T\ell) \left(1- 1_{\widetilde M_1 \times ... \times \widetilde M_{T-1}}\right) \right] \\
&\leq (C+T\ell)\E_{\Q}\left[ 1_{\widetilde{M}_1^c}(x_1)+1_{\widetilde M_2^c}(x_2)+...+1_{\widetilde M_{T-1}^c}(x_{T-1})  \right]\\
&=(C+T\ell) \left( \mu_1(\widetilde M_1^c) +  \mu_2(\widetilde M_2^c)+\dots+ \mu_{T-1}(\widetilde M_{T-1}^c) \right).
\end{split}
\end{align}
On the other hand, for any $t=2,...,T-1$,
\begin{align*}
\E_{\Q}&[ (\Delta\cdot S)_{t} 1_{\widetilde M_1 \times ... \times \widetilde M_{t-1} \times \widetilde M_{t}^c}(x_1,x_2,...,x_t) ] \\
&=\E_{\Q}\left[\left\{\Delta_1(x_2-x_1)+...+ \Delta_{t-1}(x_{t}-x_{t-1}) \right\} 1_{\widetilde M_1}(x_1)... 1_{\widetilde M_{t-1}}(x_{t-1})1_{\widetilde M^c_i}(x_{t}) \right] \\
&\geq-\E_{\Q}\left[\left\{|\Delta_1|(x_2+x_1)+\dots+ |\Delta_{t-1}|(x_{t}+x_{t-1}) \right\} 1_{\widetilde M_1}(x_1)... 1_{\widetilde M_{t-1}}(x_{t-1})1_{\widetilde M^c_i}(x_{t}) \right] \\
&\geq-\E_{\Q}\left[\left\{a_1(m_2+m_1)+...+ a_{t-1}(x_{t}+m_{t-1}) \right\} 1_{\widetilde M_1}(x_1)... 1_{\widetilde M_{t-1}}(x_{t-1})1_{\widetilde M^c_t}(x_{t}) \right]\\
&\geq-\E_{\Q}\left[\left\{a_1(m_2+m_1)+...+ a_{t-1}(x_{t}+m_{t-1}) \right\} 1_{\widetilde M^c_t}(x_{t}) \right]\\
&= -\left[a_1(m_2+m_1)+...+ a_{t-2}(m_{t-1}+m_{t-2}) + a_{t-1}m_{t-1} \right]  \mu_i(\widetilde M^c_t) - a_{t-1}\int_{\widetilde M^c_t} y d\mu_t(y),
\end{align*}
where the first inequality follows from $x_i\ge 0$ and the second inequality is due to $y<m_i$ for all $y\in \widetilde M_i$ and $|\Delta_i|\le a_i$ on $\widetilde M_i$, for all $i=1,...,T-1$. We then deduce from \eqref{d.infty(dot.bar.phi)}, \eqref{first term}, and the previous inequality that 
\begin{align}
\E_\Q[{\bar\Phi}]&\ge \E_{\Q}[\Phi+T\ell] - (C+T\ell) \left( \mu_1(\widetilde M_1^c) +  \mu_2(\widetilde M_2^c)+ \mu_{T-1}(\widetilde M_{T-1}^c) \right)\nonumber\\
&\hspace{0.2in}- \sum_{t=2}^{T-1} \bigg(\left[a_1(m_2+m_1)+...+ a_{t-2}(m_{t-1}+m_{t-2}) + a_{t-1}m_{t-1} \right]  \mu_t(\widetilde M^c_t)\nonumber\\
&\hspace{3.75in} + a_{t-1}\int_{\widetilde M^c_t} y d\mu_t(y)\bigg).\label{sum}
\end{align}
The above inequality particularly requires the linearity of outer expectations, which holds here for $\Phi+T\ell$ and $(\Phi+T\ell)1_{\widetilde M_1 \times ... \times \widetilde M_{T-1}}$. This can be proved as in the discussion below \eqref{dif:phi-phi.eps}. 
We will show that every term on the right hand side of \eqref{sum}, except $\E_{\Q}[\Phi+T\ell]$, can be made arbitrarily small, by choosing $m_t$ and $a_t$ appropriately for all $t=1,...,T-1$. 

Fix $\eps>0$, and define $\eta:= \frac{\eps}{(C+T\ell)(T-1)+(T-2)}>0$. Taking $\eps_1=\eta$ in Step 2 gives $\mu_1(M^c_1)<\eta$. Since $\mu_1(\widetilde M_1)\uparrow \mu_1(M_1)$ as $m_1\to\infty$, we can pick $m_1>1$ large enough such that $\mu_1(\widetilde M^c_1)<\eta$. With $m_1$ chosen, $a_1\ge 0$ in \eqref{a_t} is then determined.  Given the fixed $m_1$ and $a_1$, we can take $\eps_2\in (0,\eps_1)$ small enough such that $a_1m_1 \eps_2+ a_{1}\int_{A} y d\mu_2(y) <\eta$ for all $A\in\cB(\R_+)$ with $\mu_2(A)<\eps_2$. Using this $\eps_2>0$ in Step 2 gives $\mu_2(M^c_2)<\eps_2$. Since $\mu_2(\widetilde M_2)\uparrow \mu_2(M_2)$ as $m_2\to\infty$, we can pick $m_2>1$ large enough such that 
%$M_2\subseteq \R_+$ and $m_2>1$ large enough such that $\mu_2(\widetilde M^c_2)<\eps$ and 
the first term in the summation of \eqref{sum} is less than $\eta$, i.e. 
\[
a_1m_1 \mu_2(\widetilde M^c_2)+ a_{1}\int_{\widetilde M^c_2} y d\mu_2(y) <\eta.
\]
With $m_1, m_2$ chosen, $a_2\ge 0$ in \eqref{a_t} is then determined. Given the fixed $m_t$ and $a_t$ for $t=1,2$, we can take $\eps_3\in (0,\eps_2)$ small enough such that $\left[a_1(m_2+m_1)+ a_{2}m_{2} \right] \eps_3+ a_{2}\int_{A} y d\mu_3(y) <\eta$ for all $A\in\cB(\R_+)$ with $\mu_3(A)<\eps_3$. Using this $\eps_3>0$ in Step 2 gives $\mu_3(M^c_3)<\eps_3$. Since $\mu_3(\widetilde M_3)\uparrow \mu_3(M_3)$ as $m_3\to\infty$, we can pick $m_3>1$ large enough such that the second term in the summation of \eqref{sum} is less than $\eta$, i.e.
 \[
 \left[a_1(m_2+m_1)+ a_{2}m_{2} \right]  \mu_3(\widetilde M^c_3) + a_{2}\int_{\widetilde M^c_3} y d\mu_3(y)<\eta. 
 \]
By repeating the same argument for all $t=4,...,T-1$, we have $\mu_t(\widetilde M^c_t)$, $t=1,...,T-1$, and every term in summation of \eqref{sum} less than $\eta$. We then conclude from \eqref{sum} that 
\begin{equation}\label{>>>>}
\E_\Q[{\bar\Phi}]\ge \E_{\Q}[\Phi+T\ell] - ((C+T\ell)(T-1)+(T-2))\eta=\E_{\Q}[\Phi+T\ell] - \eps.    %\quad \hbox{where}\ C':= (C+T\ell)(T-1)+(T-2).
\end{equation}

{\bf Step 4:} We establish $D(\Phi)\ge D^\infty(\Phi)$. For any $\eps>0$, as \eqref{>>>>} holds for all $\Q\in\Pi$, the definition of $\widetilde P(\Phi)$ in \eqref{general duality} implies $\widetilde P(\bar\Phi)\ge \widetilde P(\Phi+T\ell)- \eps$. By Proposition~\ref{prop: D.infty}, this in turn implies $D^\infty(\bar\Phi)\ge D^\infty(\Phi+T\ell)- \eps$. Now, observe that
\[
D(\Phi)+\delta+T\ell \ge \mu(v)+T\ell =\mu(\bar v)\ge D^\infty(\bar\Phi)\ge D^\infty(\Phi+T\ell)- \eps=D^\infty(\Phi)+T\ell- \eps,
\]
where the first inequality follows from the first part of \eqref{u+(d.s)>phi}, the second inequality is due to \eqref{ine:dot.bar.phi}, and the last equality is a direct consequence of Proposition~\ref{prop: D.infty}. As $\delta,\eps>0$ are arbitrarily chosen, we conclude $D(\Phi)\ge D^\infty(\Phi)$. 
\end{proof}  
  
Thanks to Propositions~\ref{prop: D.infty} and \ref{thm: sec.main}, the proof of Theorem \ref{thm:main} is complete.

\bibliographystyle{siam}
\bibliography{refs}

\end{document}